\documentclass{article}
\usepackage{kr}

\usepackage{times}
\usepackage{soul}
\usepackage{url}
\usepackage[hidelinks]{hyperref}
\usepackage[utf8]{inputenc}
\usepackage[small]{caption}
\usepackage{graphicx}
\usepackage{amsmath}
\usepackage{amsthm}
\usepackage{booktabs}
\usepackage{algorithm}
\usepackage{algorithmic}
\usepackage{todonotes}
\newtheorem{definition}{Definition}
\newtheorem{example}{Example}
\newtheorem{theorem}{Theorem}
\newtheorem{proposition}{Proposition}
\newtheorem{lemma}{Lemma}
\newtheorem{corollary}{Corollary}
\newtheorem{remark}{Remark}

\usepackage{amssymb}
\usepackage{xspace}
\usepackage{paralist}

\usepackage{todonotes}

\usepackage{tikz}
 \usetikzlibrary{trees}
 \usetikzlibrary{shapes}
 \usetikzlibrary{fit}
 \usetikzlibrary{shadows}
 \usetikzlibrary{backgrounds}
 \usetikzlibrary{arrows,automata}
 \usetikzlibrary{positioning,fit,arrows.meta,backgrounds}

\tikzset{
    module/.style={%
        draw, rounded corners,
        minimum width=#1,
        minimum height=7mm,
        font=\sffamily
        },
    module/.default=2cm,
    >=LaTeX
}

\tikzset{
   n/.style= {circle,fill,inner sep=1.5pt,node distance=2cm}
  ,acc/.style={circle,draw,inner sep=3pt,node distance=2cm}
  ,phantom/.style={circle},
  ,arr/.style={->, >=stealth, semithick, shorten <= 3pt, shorten >= 3pt}
}

\newcommand{\LTLf}{LTL$_f$\xspace}
\newcommand{\PPLTL}{PPLTL\xspace}
\newcommand{\LTL}{LTL\xspace}

\newcommand{\LTLfp}{LTL$_f$+\xspace}
\newcommand{\PPLTLp}{PPLTL+\xspace}

\newcommand{\DFA}{{DFA}\xspace}
\newcommand{\NFA}{{NFA}\xspace}

\newcommand{\lydiasyft}{$\mathsf{LydiaSyft}$\xspace}
\newcommand{\tool}{$\mathsf{LydiaSyft+}$\xspace}

\title{Emerson-Lei and Manna-Pnueli Games for \LTLfp and \PPLTLp Synthesis
}

\author{%
Daniel Hausmann$^1$\and
Shufang Zhu$^1$\and
Gianmarco Parretti$^2$\and
Christoph Weinhuber$^3$\\
Giuseppe {De Giacomo}$^{2,3}$\and
Nir Piterman$^4$\\
\affiliations
$^1$University of Liverpool, UK\qquad
$^2$Sapienza University of Rome, Italy\qquad
$^3$University of Oxford, UK\\
$^4$University of Gothenburg and Chalmers University of Technology, Sweden\\
\emails
\{shufang.zhu, hausmann\}@liverpool.ac.uk,
parretti@diag.uniroma1.it,\\
\{christoph.weinhuber, giuseppe.degiacomo\}@cs.ox.ac.uk,
piterman@chalmers.se
}

\begin{document}
\maketitle
\begin{abstract}
Recently, the Manna-Pnueli Hierarchy has been used to define the temporal logics \LTLfp and \PPLTLp, which allow to use finite-trace \LTLf/\PPLTL techniques in infinite-trace settings while achieving the expressiveness of full \LTL. 
In this paper, we present the first actual solvers for reactive synthesis in these logics. These are based on games on graphs that leverage DFA-based techniques from \LTLf/\PPLTL to construct the game arena. We start with a symbolic solver based on Emerson-Lei games, which reduces lower-class properties~(guarantee, safety) to higher ones (recurrence, persistence) before solving the game.
We then introduce Manna-Pnueli games, which natively embed Manna-Pnueli objectives into the arena. These games are solved by composing solutions to a DAG of simpler Emerson-Lei games, resulting in a provably more efficient approach.
We implemented the solvers and practically evaluated their performance on a range of representative formulas. The results show that Manna-Pnueli games often offer significant advantages, though not universally, indicating that combining both approaches could further enhance practical performance.
\end{abstract}

\section{Introduction}\label{sec:intro}

This paper is about devising actual solvers for reactive synthesis in \LTLf{+} and PPLTL+, which can be seen as a Manna-Pnueli normal form for Linear Temporal Logic (\LTL) based respectively on its finite trace variant \LTLf and on Pure Past LTL  (PPLTL).

Reactive synthesis deals with synthesizing programs (aka strategies) from temporal specifications, for systems
(e.g. agents, processes, protocols, controllers, robots) that interact with their environments during their execution \cite{PnueliR89,finkbeiner2016synthesis,EhlersLTV17}. The most common specification language is possibly Linear Temporal Logic (LTL)~\cite{Pnueli77}.
Reactive synthesis shares foundational techniques with model checking, grounded in the interplay between logic, automata, and games \cite{fijalkow2023games}. For LTL, synthesis normally proceeds by: (1) specifying the desired behavior with controllable and uncontrollable variables; (2) translating the specification into an equivalent automaton over infinite words; (3) determinizing the automaton—unlike in model checking—to define a game between system and environment; and (4) solving the game, typically with a parity objective, to derive a strategy satisfying the original specification.

The symbolic techniques for model checking based on Boolean encodings to compactly represent game arenas and compute fixpoints can also be leveraged in reactive synthesis.
Nevertheless, differently from model checking, LTL synthesis has struggled to achieve comparable efficiency, primarily due to the inherent complexity of Step (3): the determinization of nondeterministic Büchi automata~(NBA), a process known to be computationally challenging~\cite{Vardi07,AlthoffTW06}.


In AI, reactive synthesis is closely related to strong planning for temporally extended goals in fully observable nondeterministic domains \cite{Cimatti03,DBLP:journals/amai/BacchusK98,BacchusK00,CalvaneseGV02,BaierFM07,GereviniHLSD09,DR-IJCAI18,CamachoBM19}.
Plans are typically assumed to terminate, and
this has led to a focus on logics over finite traces, rather than infinite ones, with \LTLf, the finite-trace variant of LTL, being a common choice \cite{GPSS80,BaierM06,DegVa13,DegVa15}.
In fact, \LTLf synthesis \cite{DegVa15} is, along with the GR(1) fragment of LTL \cite{PiPS06}, one of the two major success stories in reactive synthesis to date.

The steps of the \LTLf synthesis algorithm closely mirror Steps (1)–(4) outlined for LTL, potentially incurring similar asymptotic blowups. Specifically, Step (2) yields a nondeterministic finite automaton (NFA) that can be exponentially larger than the input formula, and Step (3) determinizes it into a DFA, which can be exponentially larger than the NFA. However, in practice, these worst-case blowups rarely occur for \LTLf:  Step (2) is shared with model checking, where it performs well; Step (3) benefits from the fact that NFA determinization is rarely problematic in practice -- in fact, it is often observed that the resulting DFA is actually smaller than the original NFA \cite{TaVa05,ArmoniKTVZ06,RozierV12,TabakovRV12,zhu21}.

Moreover, tools like MONA \cite{KlarlundMS02}, used by \LTLf synthesizers to extract  automata, return DFAs in semi-symbolic form, i.e., transitions are represented symbolically. This is a very important aspect, because one other significant barrier that we may underestimate in LTL and \LTLf synthesis is the alphabet explosion problem. That is to say, once there are just 20 propositions, which is not a very large number, the alphabet already has size 1 million, making explicit treatment of the alphabet infeasible. Symbolic
methods circumvent this problem (to some extent).

These nice characteristics of \LTLf have enabled the community to build a first fully symbolic solver for \LTLf synthesis \cite{ZTLPV17}. Since then, a series of papers have improved symbolic technology for \LTLf synthesis significantly \cite{bansal2020hybrid,DF2021,KankariyaB24,lydiasyft}. Furthermore, the \LTLf technology has been extended to several other cases, from safety and reachability \cite{ZhuTLPV17,BansalGSLVZ22,AminofGSFRZ25} all the way to GR(1)~\cite{DeGiacomoSTVZ22}.

These desirable properties are also exhibited by Pure-Past LTL (PPLTL) \cite{LichtensteinPZ85,DeGiacomoSFR20}, which has recently emerged as an even simpler alternative to \LTLf, while maintaining the same expressive power \cite{BonassiGFFGS23icaps,BonassiGFFGS23ecai,BonassiDGS24ecai}.

The question is, can we use this technology to do LTL synthesis? Recently  \cite{LTLfplus} gave affirmative answer to this question by exploiting Manna and Pnueli’s normal form \cite{DBLP:conf/podc/MannaP89}. 
This normal form is based on specifying finite trace properties, e.g. expressed in \LTLf (or PPLTL) and then requiring that the property holds for \emph{some} prefixes of infinite traces (\emph{guarantee} properties), for \emph{all} prefixes (\emph{safety} properties), for \emph{infinitely many} prefixes (\emph{recurrence} properties), or for \emph{all but finitely many} prefixes (\emph{persistence} properties). Any LTL formula can be expressed as a Boolean combination of these four classes. 

How can we take advantage of Manna and Pnueli’s normal form? An answer is, first of all, in building the arena. 
For each finite trace property, we can build the DFA, using \LTLf technology \cite{LTLfplus}. Then, depending on Manna and Pnueli’s class, we choose the accepting condition to (guarantee) visit a final state once, (safety) never leave the set of final states, (recurrence) visit some final state infinitely often, (persistence) stay out of the set of final states finitely many times. All these component automata run in parallel, so we can take their product. In fact, we can represent the product symbolically in a straightforward way.
To solve these temporal conditions, we have several options. With a simple (polynomial) manipulation, we can transform  guarantee and safety properties, and corresponding automata, into recurrence properties, see \cite{LTLfplus}. What we get is an Emerson-Lei game where the only temporal conditions are boolean combinations of recurrence and persistence \cite{EmersonL87}. Such a game can be solved symbolically through a fixpoint algorithm based on compact semantic representations of Emerson-Lei objectives, called Zielonka Trees \cite{DBLP:conf/fossacs/HausmannLP24}.
Note that Emerson-Lei games can be reduced (with an exponential multiplicative factor) to parity games, for which there are well-developed solvers. However, while in \cite{LTLfplus} it was shown that no further determinization would be needed to generate the parity automaton, we do need to handle possible permutations (latest appearance records) which destroy the symbolic representation. Instead, the Emerson-Lei approach preserves the symbolic representation, which is used for the fixpoint computation directly.
We stress that  both these solutions are worst-case optimal and solve  \LTLf{+} synthesis in 2EXPTIME  and PPLTL+ synthesis in EXPTIME, respectively (these problems are 2EXPTIME-complete and EXPTIME-complete) \cite{LTLfplus}.  However, there is a sharp difference in practical terms among them, since only the Emerson-Lei based solution preserves the symbolic structure coming from the finite-trace properties.

Can we improve on the Emerson-Lei approach? In particular, can we avoid reducing guarantee and safety, which do not require nested fixpoints, to recurrence, which indeed requires nesting and complicates the Zielonka Tree? In this paper we answer affirmatively. We introduce a new kind of games called \emph{Manna-Pnueli games} that handle the combination of conditions guarantee, safety, recurrence and persistence directly. In particular, for such games, we give a symbolic fixpoint solution analogous to that of Emerson-Lei games in \cite{DBLP:conf/fossacs/HausmannLP24}, but exploiting the simplicity of guarantee and safety conditions to simplify the fixpoint computation.

The contributions of this paper are as follows:
\begin{compactitem}
\item 
First, we present a symbolic synthesizer based on Emerson-Lei games, which reduces lower-class properties (guarantee, safety) to higher-class ones (recurrence, persistence) before solving the game.

\item 
We then introduce Manna-Pnueli games, which natively deal with Manna-Pnueli objectives on the game arena.

\item 
Next, we present a symbolic synthesizer based on these new Manna-Pnueli games.

\item 
We show that Manna-Pnueli games can be solved by composing solutions to a DAG of simpler Emerson-Lei games.

\item 
We prove that this compositional approach to solving Manna-Pnueli games is asymptotically more efficient than the naïve reduction to Emerson-Lei games.

\item 
Finally, we implement both solvers using state-of-the-art symbolic technology and evaluate their performance on a range of representative formulas.
\end{compactitem}
The results show that Manna-Pnueli games often offer significant advantages, though not universally, suggesting that combining both approaches could further enhance practical performance.


\section{Preliminaries}\label{sec:prelim}

\paragraph{\LTLfp and \PPLTLp.}\label{par:logics}

We briefly recall the logics \LTLfp and \PPLTLp following \cite{LTLfplus}.
These logics allow to express Boolean combinations of guarantee, safety, recurrence, and persistence properties over
\emph{finite traces}; the underlying finite trace properties are specified either in \LTLf~\cite{DegVa13}, that is, in \LTL evaluated over finite traces,
or in \PPLTL, that is, in pure past LTL~\cite{DeGiacomoSFR20}. Note that in this paper we use $\mathsf{X}$ for weak next and $\mathsf{X[!]}$ for strong next.

Formulas of \LTLfp (resp. \PPLTLp) over a countable set $AP$ of propositions are constructed by the grammar
\begin{align*}
\Psi, \Psi' ::= \forall\Phi \mid \exists\Phi \mid \forall\exists\Phi \mid \exists\forall\Phi \mid \Psi \lor \Psi' \mid \Psi \land \Psi' \mid \neg \Psi
\end{align*}
where the $\Phi$ are finite trace \LTLf (resp. \PPLTL) formulas.
We use $[\Phi]\subseteq (2^{AP})^*$ to
denote the set of finite traces over $2^{AP}$ that satisfy a finite trace formula $\Phi$. 
Given a set $T\subseteq (2^{AP})^*$ of finite traces, we let $\exists T$ ($\forall T$) denote the set of infinite traces $\tau\in (2^{AP})^\omega$
such that at least one finite prefix of $\tau$ is (all finite prefixes of $\tau$ are) contained in $T$;
similarly, we let $\forall\exists T$ ($\exists\forall T$) denote the set of infinite traces for which infinitely many (all but finitely many) 
prefixes are contained in $T$.
Then we evaluate \LTLfp/\PPLTLp formulas $\Psi$ over infinite traces using the extension $[\Psi]\subseteq (2^{AP})^\omega$
defined inductively by
$[\Psi\lor\Psi'] = [\Psi]\cup [\Psi']$, 
$[\Psi\land\Psi'] = [\Psi]\cap [\Psi']$,
$[\neg\Psi] = (2^{AP})^\omega\setminus[\Psi]$, $[\mathbb{Q}\Phi] = \mathbb{Q} [\Phi]$
where $\mathbb{Q}\in\{\exists,\forall,\forall\exists,\exists\forall\}$.
The logics \LTLf, \PPLTLp, and \LTL define the same infinite-trace properties.

\smallskip
\noindent\textbf{Automata on finite traces.}\label{par:automata}
A transition system $T=(\Sigma,Q,I,\delta)$ consists
of a finite alphabet $\Sigma$, a finite set $Q$ of states, a set $I\subseteq Q$
of initial states, and a transition relation $\delta\subseteq Q\times \Sigma\times Q$.
For $q\in Q$ and $a\in\Sigma$, we define $\delta(q,a)=\{q'\in Q\mid (q,a,q')\in \delta\}$.
A transition system is \emph{deterministic} if $|I|=1$ and $|\delta(q,a)|=1$ for all
$q\in Q$ and $a\in\Sigma$, and
\emph{nondeterministic} otherwise; for deterministic transition systems, we write $\delta(q,a)=q'$ where
$q'\in Q$ is \emph{the} state such that $q'\in\delta(q,a)$ and denote the initial state by $\iota$.
A \emph{finite automaton} $\mathcal{A}=(T,F)$ is a transition system together with a set $F\subseteq Q$ of accepting states.
If $T$ is deterministic, then $\mathcal{A}$ is a deterministic finite automaton (\DFA), otherwise it is a nondeterministic finite automaton (\NFA).
A \emph{run} of an automaton on a word $w\in \Sigma^*$ is a path through $T$ starting at an initial state such that the sequence of transition labels of the path is $w$;
a finite run is accepting if it ends in an accepting state.
An automaton \emph{accepts} the language $L(\mathcal{A})$ consisting of all finite words for which
there is an accepting run of $\mathcal{A}$.

Finite trace \LTLf and \PPLTL formulas over $AP$ can be turned into equivalent finite automata (with alphabet $2^{AP}$), accepting
exactly the traces that satisfy the formulas.
For each \LTLf formula $\varphi$, there is an equivalent \NFA of size $2^{\mathcal{O}(|\varphi|)}$ and an equivalent \DFA of size $2^{2^{\mathcal{O}(|\varphi|)}}$ \cite{DegVa15}.
For each \PPLTL formula $\varphi$, there is an equivalent \DFA of size $2^{\mathcal{O}(|\varphi|)}$ \cite{DeGiacomoSFR20}.

\smallskip
\noindent\textbf{Infinite-duration games on finite graphs.}\label{par:games}
A \emph{game arena} is a finite directed graph $A=(V,E\subseteq V\times V)$ such that $V$ is partitioned into the sets $V_s$ and $V_e$ of game nodes controlled by
the \emph{system player} and the \emph{environment player}, respectively.  
Define $E(v)=\{v'\in V\mid (v,v')\in E\}$ for $v\in V$ and assume that $E(v)\neq \emptyset$ for all $v\in V$.
A \emph{play} is a path in $A$; let $\mathsf{plays}(A)$ denote the set of infinite plays over $A$.
A \emph{strategy} for the system player is a function $\sigma:V^*\cdot V_s \to V$ that assigns a single game node $\sigma(v_0 v_1 \ldots v_n)\in E(v_n)$ to any finite
play $v_0 v_1 \ldots v_n\in V^*\cdot V_s$ that ends in a game node owned by the system player ($v_n\in V_s$). A play $v_0 v_1 \ldots$
is \emph{compatible} with a strategy if for every $i$ such that
$v_i\in V_s$, $v_{i+1}=\sigma(v_0 v_1 \ldots v_i)$.
An \emph{objective} on an arena $A$ is a set $O\subseteq \mathsf{plays}(A)$ of plays; then a play is \emph{winning}
for the system player if it is contained in $O$.
Notions of strategies and winning plays for the environment player
are defined dually.
A strategy $\sigma$ wins a node $v\in V$ for a player
if all plays that are compatible with $\sigma$ and start at $v$ are winning for that player.
\emph{Solving} a game amounts to computing the set of game nodes won by the two players, together with their witnessing strategies.

\smallskip
\noindent\textbf{Reactive synthesis and games.}\label{par:syntgames}
In the context of reactive synthesis, we assume that the set $AP$ of atomic propositions is
partitioned into sets $X$ and $Y$, denoting \emph{system} and \emph{environment actions}, respectively.
A \emph{(synthesis) strategy} is a function $\sigma:(2^Y)^*\to 2^X$, and an \emph{outcome}
of such a strategy is an infinite word $(x_0\cup y_0)(x_1\cup y_1)\ldots\in (2^{AP})^\omega$
such that $x_{i+1}=\sigma(y_0 y_1\ldots y_i)$ for all $i\geq 0$.
Thus, synthesis strategies encode the behavior of \emph{transducers}.
\begin{definition}~\cite{LTLfplus}
The \emph{synthesis problem} for \LTLfp (resp. \PPLTLp) asks for a
given formula $\Psi$ whether there is a strategy $\sigma$ such that every 
outcome of $\sigma$ satisfies $\Psi$, and if so, to return such a strategy.
\end{definition}

Then a deterministic transition system $T=(2^{\mathsf{AP}},Q,\iota,\delta)$ induces
a game arena $A_T=(Q\cup Q\times 2^X\cup Q\times 2^X\times 2^Y,E)$ with the system player owning nodes $q\in Q$
and the environment player owning all other nodes; the moves are defined by putting $E(q)=\{q\}\times 2^X$,
$E(q,x)=\{(q,x)\}\times 2^Y$, and  $E(q,x,y)=\{\delta(q,x\cup y)\}$.
Plays $q_0(q_0,x_0)(q_0,x_0,y_0)q_1(q_1,x_1)(q_1,x_1,y_1)\ldots$ over $A_T$ induce runs
$q_0 q_1\ldots$ of $T$ on words $(x_0\cup y_0)(x_1\cup y_1)\ldots$.

Hence the synthesis problem can be solved by transforming the input formula
into a deterministic transition system and then solving the induced game
with a suitable objective.

\section{Synthesis via EL Games}\label{sec:ELsynthesis}

It has been shown~\cite{LTLfplus} that the synthesis problem of \LTLfp (resp. \PPLTLp) reduces
to the solution of games with so-called Emerson-Lei objectives. These are
Boolean combinations of recurrence and persistence objectives, formally defined as follows.

Given a finite set $\Gamma$ of events,
an \emph{Emerson-Lei (EL) formula} (over $\Gamma$) is a positive Boolean formula
over atoms of the shape 
	$\mathsf{GF}\,a$ or $\mathsf{FG}\,a$, where $a\in \Gamma$.
We evaluate EL formulas over infinite sequences
of sets of events, that is, over elements of $(2^\Gamma)^\omega$. Given $L_1 L_2\ldots\in (2^\Gamma)^\omega$, we put 
\begin{align*}
L_1 L_2\ldots&\models \mathsf{GF}\,a &\Leftrightarrow && \forall i.\,\exists j>i.\, a\in L_j\\
L_1 L_2\ldots&\models \mathsf{FG}\,a &\Leftrightarrow && \exists i.\,\forall j>i.\, a\in L_j
\end{align*}
The evaluation of Boolean combinations ($\land,\lor$) of EL formulas and of Boolean constants ($\top$,$\bot$) is as expected.
Given a finite set $Q$, an infinite sequence $\pi=q_0 q_1 \ldots \in Q^\omega$,
a labeling function $\gamma:Q\to 2^{\Gamma}$,
and an EL formula $\varphi$, we denote
$\gamma(q_0)\gamma(q_1)\ldots\models \varphi$ by $\pi\models\varphi$.

An \emph{Emerson-Lei objective} $O=(\Gamma,\gamma,\varphi)$ on a finite set $Q$ is given in the form of a set $\Gamma$ of events, a labeling function
$\gamma:Q\to 2^\Gamma$ and an EL formula $\varphi$ over $\Gamma$; we generally assume that $|\Gamma|\leq |Q|$. A sequence $\pi\in Q^\omega$ is 
contained in $O$ (by slight abuse of notation) if and only if $\pi\models\varphi$.

An \emph{Emerson-Lei automaton} $\mathcal{A}=(T,O)$ is a transition
system $T$ together with an EL objective on $T$. Then $\mathcal{A}$ recognizes
the ($\omega$-regular) language $L(\mathcal{A})$ of all infinite words
for which there is run $\pi$ of $T$ such that $\pi\in O$.

An \emph{Emerson-Lei game} $G=(A,O)$ consists of a game arena $A=(V,E)$
together with an EL objective $O$ on $A$, specifying the winning
plays. EL games are \emph{determined}, that is, every game node
is won by exactly one of the players.

\begin{theorem}\label{thm:ELsolution}\cite{DBLP:journals/apal/McNaughton93,DBLP:journals/tcs/Zielonka98}
Emerson-Lei games with $n$ nodes and $k$ Emerson-Lei events can be solved in time $\mathcal{O}(k!\cdot n^{k+2})\in 2^{\mathcal{O}(k\log{n})}$;
winning strategies require at most $k!$ memory values.
\end{theorem}

Recently, a symbolic algorithm for the solution of EL games has been proposed~\cite{DBLP:conf/fossacs/HausmannLP24}. This algorithm
leverages \emph{Zielonka trees}, that is, succinct semantic representations of EL objectives, to transform EL objectives into equivalent
fixpoint equation systems. 
EL game solution then can be performed by symbolic solution of the corresponding equation system.
The resulting algorithm realizes the time bound stated in Theorem~\ref{thm:ELsolution}.

\begin{remark}\label{rem:ELstrat}
Strategy extraction for EL games works as described in~\cite{DBLP:conf/fossacs/HausmannLP24}.
Leaves in Zielonka trees are labelled with sets of events, thought of as a memory of recently visited events.
Winning strategies then use leaves in Zielonka trees as memory values and combine the events of current game
nodes with the events in the current memory values to update the set of recently visited events
and obtain new memory values. 
Strategies are based 
on \emph{existential subtrees} of Zielonka trees, reducing strategy sizes~\cite{DBLP:conf/lics/DziembowskiJW97}.
\end{remark}

Next, we show how this symbolic solution algorithm for EL games can be used for \LTLfp (resp. \PPLTLp) synthesis.

Consider an input \LTLfp or \PPLTLp formula $\Psi$ given in positive normal form, that is, given as a positive Boolean formula over $k$ atoms
$\mathbb{Q}_i\Phi_i$ where $\mathbb{Q}_i\in\{\exists,\forall,\exists\forall,\forall\exists\}$,
and where all $\Phi_i$ are \LTLf (resp. \PPLTL) formulas.
The synthesis algorithm transforms $\Psi$ into an equivalent EL automaton and then solves the EL game induced by this automaton.
\vspace{-10pt}

\paragraph{Step 1.} For each $i\in[k]$, convert the finite trace formula $\Phi_i$ into an equivalent DFA $(D_i, F_i)$ where
$D_i=(2^{AP},Q_i,\iota_i,\delta_i)$.
Assume without loss of generality that $\iota_i$ does not have incoming transitions.
If $\mathbb{Q}_i=\forall$, then add $\iota_i$ to $F_i$
and turn every non-accepting state into a non-accepting sink state
(for $q\notin F_i$ and all $a\in\Sigma$, put $\delta_i(q,a)=q$); if $\mathbb{Q}_i=\exists$, then remove $\iota_i$ from $F_i$ and turn every accepting state into an accepting sink state (for $q\in F_i$ and all $a\in\Sigma$, put $\delta_i(q,a)=q$).
Construct the product transition system $D_\Psi=\prod_{i\in[k]} D_i$ and let $Q_\Psi$ denote the state space of $D_\Psi$.
\smallskip

We point out that turning (non)accepting states into sinks in Step 1. does not increase the size of automata. For $i$ such that $\mathbb{Q}_i=\forall$ or $\mathbb{Q}_i=\exists$,
the resulting automata intuitively incorporate memory on whether a (non)accepting state has been visited so far, or not. The following is immediate.

\begin{lemma}\label{lem:localMemory}
Let $\mathbb{Q}_i=\exists$ (resp. $\mathbb{Q}_i=\forall$). Then an infinite path in $D_i$ 
visits $F_i$ (resp. $Q_i\setminus F_i$) at least once if and only if
the run eventually visits only states from $F_i$ (resp. $Q_i\setminus F_i$).
\end{lemma}
For $q=(q_1,\ldots,q_k)\in Q$, let $L_q=\{i\in [k]\mid q_i\in F_i, \mathbb{Q}_i\in \{\forall,\exists\}\}$
denote the local events for which the corresponding automaton is in an accepting state in $q$.
\vspace{-10pt}

\paragraph{Step 2.} 
Define the EL objective $O=(\Gamma,\gamma,\varphi)$ by putting $\Gamma=[k]$ and $\gamma(q_1,\ldots,q_k)=\{i\in[k]\mid q_i\in F_i\}$
for $(q_1,\ldots,q_k)\in Q_\Psi$.
Depending on the shape of $\mathbb{Q}_i$, define
$\varphi_i$ to be $\mathsf{GF}\,i$ (if $\mathbb{Q}_i\in\{\exists,\forall,\forall\exists\}$), or $\mathsf{FG}\,i$ (if $\mathbb{Q}_i=\exists\forall$).
The EL formula $\varphi$ is obtained from the input formula $\Psi$ by replacing each atom $\mathbb{Q}_i \Phi_i$ with $\varphi_i$.
Define the deterministic EL automaton $\mathcal{A}_\Psi=(D_\Psi,O)$.\smallskip

\begin{proposition}~\cite{LTLfplus}
The formula $\Psi$ is equivalent to the EL automaton $\mathcal{A}_\Psi$.
\end{proposition}
\vspace{-10pt}
\paragraph{Step 3.} Solve the EL game induced by $\mathcal{A}_\Psi$. If the system player wins the initial state 
$(\iota_1,\ldots,\iota_k)$ of $D_\Psi$,
then extract a witnessing strategy.


\begin{theorem}
The \LTLfp (resp. \PPLTLp) synthesis problem 
can be decided symbolically via EL games in \textsc{2EXPTIME} (resp. \textsc{EXPTIME}).
\end{theorem}
In more detail, consider an input formula $\Psi$
of size $n$ and consisting of $k$ recurrence and persistence formulas and 
$d$ guarantee and safety formulas.
The constructed EL game over $\mathcal{A}_\Psi$ is of size at most
$n'=2^{2^n}$ (resp. $n'=2^n$) and the objective $O$ has $k+d$ events. 
By Theorem~\ref{thm:ELsolution}, it can be solved in time
$2^{\mathcal{O}((k+d)\log{n'})}$.

\begin{remark}
All steps in the described synthesis algorithm are open to symbolic implementation.

The construction
of individual DFAs for finite trace subformulas in Steps 1. and 2. is based on
the powerset construction which is amenable to symbolic implementation, and taking the product of the individual automata
in Step 3. is an inherently symbolic operation. Finally, the game solution
in Step 4. can be implemented using the symbolic algorithm for EL game solution
from~\cite{DBLP:conf/fossacs/HausmannLP24}.
\end{remark}

\begin{remark}
Transducers for realizable specifications $\Psi$ are obtained by extracting a winning strategy 
for the system player in the EL game induced by $\mathcal{A}_\Psi$ (see Remark~\ref{rem:ELstrat}).
\end{remark}

\section{MP Automata and Games}\label{par:ELMPgames}

Next, we introduce automata and games with objectives that support, in addition to the Boolean combinations of recurrence and persistence allowed in EL objectives,
also combinations with guarantee and safety objectives. We call such objectives Manna-Pnueli (MP) objectives.
Adding guarantee and safety properties to EL objectives enables a direct translation from \LTLfp (resp. \PPLTLp) formulas to automata, as the structure of the resulting objectives directly corresponds to the high-level structure of formulas. In Section~\ref{sec:synthesis} below, MP automata and 
the solution of MP games will be instrumental
to an alternative solution of the synthesis problem for \LTLfp (resp. \PPLTLp).

Given a finite set $\Gamma$ of events,
a \emph{Manna-Pnueli formula} is a positive Boolean formula
over atoms of the shape 
	$\mathsf{GF}\,a$, $\mathsf{FG}\,a$,
	$\mathsf{F}\,a$, or $\mathsf{G}\,a$, where $a\in \Gamma$.
    We refer to events that occur in a concrete formula only in atoms of the shape
    $\mathsf{F}\,a$ or $\mathsf{G}\,a$ as \emph{local events}, and
    as \emph{Emerson-Lei events} to all other events.
Given $L_1 L_2\ldots\in (2^\Gamma)^\omega$, we put 
\begin{align*}
L_1 L_2\ldots&\models \mathsf{F}\,a &\Leftrightarrow && \exists i.\, a\in L_i\\
L_1 L_2\ldots&\models \mathsf{G}\,a &\Leftrightarrow && \forall i. \,a\in L_i
\end{align*}
All other operators are evaluated in the same way as for EL formulas;
denote $\gamma(q_0)\gamma(q_1)\ldots\models \varphi$ by $\pi\models\varphi$,
where $\gamma:Q\to 2^\Gamma$ is a labelling function, and $\pi=q_0 q_1 \ldots \in Q^\omega$.

A \emph{Manna-Pnueli objective} $O=(\Gamma,\gamma,\varphi)$ on a finite set $Q$ consists of a set $\Gamma$ of events, a labeling function
$\gamma:Q\to 2^\Gamma$ and an MP formula $\varphi$ over $\Gamma$; assume that $|\Gamma|\leq |Q|$. A sequence $\pi\in Q^\omega$ is 
contained in $O$ if and only if $\pi\models\varphi$.

We point out that like EL objectives, MP objectives are $\omega$-regular (that is, they can be transformed to parity objectives); however, unlike
EL objectives, MP objectives are not prefix independent
(that is, there may be $u_1,u_2\in Q^*$, $v\in Q^\omega$ such that $u_1 v\in O$ but $u_2 v\notin O$).

A \emph{Manna-Pnueli automaton} $\mathcal{A}=(T,O)$ is a transition
system $T$ with set $Q$ of states together with a Manna-Pnueli objective on $Q$. Then $\mathcal{A}$ recognizes
the ($\omega$-regular) language $L(\mathcal{A})$ consisting of all infinite words
for which there is a run $\pi$ of $T$ such that $\pi\in O$.

A \emph{Manna-Pnueli game} $G=(A,O)$ consists of a game arena $A=(V,E)$
together with a Manna-Pnueli objective $O$ on $A$. By definition, MP games are determined.


\begin{example}\label{ex:MP}
Consider the game $G$ with events $a,b,c,d$, with Manna-Pnueli objective
$\varphi=(\mathsf{GF}\,a\land\mathsf{GF}\,b\land\mathsf{G}\,d)\lor (\mathsf{FG}\,\neg b\land \mathsf{F}\,c)$ and
with node ownership indicated
by circles (system player) and boxes (environment player).

    In this example the system player wins every node by a
strategy that uses memory to alternatingly move from the 
node labelled with $d$ to the nodes labelled with $a,d$ and
$b,d$, respectively; the strategy always moves from the node
labelled with $a$ to the node labelled with $c$.

\begin{center}
\scalebox{0.8}{
\tikzset{every state/.style={minimum size=15pt}, every node/.style={minimum size=15pt}}
    \begin{tikzpicture}[
		auto,
    node distance=1.7cm,
    thick
    ]
    \tikzstyle{every state}=[
        draw = black,
        thick,
        fill = white,
        minimum size = 4mm
    ]
    
     \node[state, label={$d$}] (1)  {};
     \node[state, rectangle,label={$c,d$}] [left of=1] (2) {};
     \node[state,label={left:$a,d$}] [below of=1] (3) {};
     \node[state, rectangle,label={$b,d$}] [right of=1] (4) {};
     \node[state,label={left:$a$}] [below of=4] (5)  {};
     \node[state,label={$c$}] [right of=5] (6)  {};

     \node[label={$G$}] (0) at (-2.6,0.4) {};

     \path[->] (1) edge[bend left=15] node {} (2);
     \path[->] (2) edge[bend left=15] node {} (1);

     \path[->] (1) edge[bend left=15] node {} (3);
     \path[->] (3) edge[bend left=15] node {} (1);

     \path[->] (1) edge[bend left=15] node {} (4);
     \path[->] (4) edge[bend left=15] node {} (1);

     \path[->] (4) edge[bend left=15] node {} (5);
     \path[->] (5) edge[bend left=15] node {} (4);

     \path[->] (5) edge[bend left=15] node {} (6);
     \path[->] (6) edge[bend left=15] node {} (5);

     \path[->] (2) edge[bend left=45] node {} (4);

  \begin{pgfonlayer}{background}
    \filldraw [line width=4mm,join=round,black!10]
      (-2,-2)  rectangle (3.7,0.9);
  \end{pgfonlayer}

  \end{tikzpicture}
  }

\end{center}

Any play
following this strategy either avoids the node labelled with just $a$ forever and infinitely often visits the nodes labeled with $a,d$ and $b,d$
(and then satisfies $\mathsf{GF}\,a\land\mathsf{GF}\,b\land\mathsf{G}\,d$),
or it eventually only visits the two bottom right nodes
(and then satisfies $\mathsf{FG}\,\neg b\land \mathsf{F}\,c$).
\end{example}

\section{Synthesis via MP Games}\label{sec:synthesis}
We now show how \LTLfp and \PPLTLp synthesis both reduce to the solution of MP games.
To this end, we consider an input \LTLfp or \PPLTLp formula $\Psi$ given in positive normal form over $k$ finite trace formulas,
as in Section~\ref{sec:ELsynthesis}.
The synthesis algorithm transforms $\Psi$ to an equivalent deterministic MP automaton and solves the MP game induced by the automaton.
We point out that the transformation from $\Psi$ to an MP automaton is immediate (and does, unlike the construction in
Section~\ref{sec:ELsynthesis}, not transform any states into sink states), as MP automata natively embed the
high-level structure of \LTLfp and \PPLTLp formulas. 
\vspace{-10pt}

\paragraph{Step 1a.} Convert the finite trace formulas $\Phi_i$ into equivalent DFAs $(D_i,F_i)$ with $D_i=(2^{AP},Q_i,\iota_i,\delta_i)$.
 If $\mathbb{Q}_i=\forall$, then add $\iota_i$ to $F_i$; if $\mathbb{Q}_i=\exists$, then remove $\iota_i$ from $F_i$.
 Construct the product transition system $D'_\Psi=\prod_{i\in[k]} D_i$.
\vspace{-10pt}

\paragraph{Step 2a.}
Define the Manna-Pnueli objective $O'=(\Gamma,\gamma,\varphi)$ by putting $\Gamma=[k]$ and $\gamma(q_1,\ldots,q_k)=\{i\in[k]\mid q_i\in F_i\}$
for $(q_1,\ldots,q_k)\in Q$.
Depending on the shape of $\mathbb{Q}_i$, define
$\varphi_i$ to be $\mathsf{F}\,i$ (if $\mathbb{Q}_i=\exists$), $\mathsf{G}\,i$ (if $\mathbb{Q}_i=\forall$), $\mathsf{GF}\,i$  (if $\mathbb{Q}_i=\forall\exists$), or $\mathsf{FG}\,i$ (if $\mathbb{Q}_i=\exists\forall$).
The Manna-Pnueli formula $\varphi$ is obtained from the input formula $\Psi$ by replacing each atom $\mathbb{Q}_i \Phi_i$ with $\varphi_i$.
Define the deterministic Manna-Pnueli automaton $\mathcal{A}'_\Psi=(D'_\Psi,O')$.
\vspace{-10pt}

\paragraph{Step 3a.} Solve the Manna-Pnueli game induced by $\mathcal{A}'_\Psi$. If the system player wins the initial state 
$(\iota_1,\ldots,\iota_k)$ of $D'_\Psi$,
then extract a witnessing strategy.\smallskip

We show correctness of the simpler automata construction in Step 1a.
To this end, define the objective $O_i=(\Gamma_i,\gamma_{i},\varphi_i)$ on $D_i$, where $\Gamma_i=[1]$ and $\gamma_{i}:D_i\to[1]$ maps
nodes from $F_i$ to $\{1\}$, and all other nodes to $\emptyset$.
Then showing the following lemma is immediate.
\begin{lemma}\label{lem:DMPAconstruction}
Let $\pi$ be an infinite trace. For all $i\in[k]$, we have
$\pi\models \mathbb{Q}_i\Phi_i$ if and only if
$\pi\in L(D_i,O_i)$.
\end{lemma}

\begin{proposition}\label{prop:formulaToDMPA}
The formula $\Psi$ is equivalent to the MP automaton $\mathcal{A}'_\Psi$.
\end{proposition}
The proof is by induction over $\Psi$, using Lemma~\ref{lem:DMPAconstruction} for the base cases.

\begin{theorem}
The \LTLfp (resp. \PPLTLp) synthesis problem 
can be decided symbolically via MP games in \textsc{2EXPTIME}~(resp. \textsc{EXPTIME}).
\end{theorem}

In more detail, consider an input formula $\Psi$
of size $n$ and with $k$ recurrence and persistence formulas, and $d$ guarantuee and safety formulas.
The constructed game over $\mathcal{A}'_\Psi$ is of size at most
$n'=2^{2^n}$ (resp. $n'=2^n$) and has $k$ EL events and $d$ local events.

While the MP automata construction described in Step 1a. is natural and straightforward,
it has the drawback that the resulting automata do not incorporate memory for local events.
A more efficient procedure is obtained by using the automata construction from Section~\ref{sec:ELsynthesis} (Step 1.) that 
adds sink states and incorporates memory for local events at no extra cost (Lemma~\ref{lem:localMemory}).
In what follows, we assume that the latter automata construction is used in the construction of $\mathcal{A}'_\Psi$.
Then the synthesis
game can be solved in time
$2^{\mathcal{O}(k\log{n'})}$ by reduction to a composition of EL games (as described in Section~\ref{sec:ELreduction2} below, see Theorem~\ref{thm:MPDAGsolution}).
This indicates that the proposed synthesis method can handle guarantee and safety properties ``for free''. 
It improves significantly over the reduction to EL games from Section~\ref{sec:ELsynthesis}
which solves the
synthesis problem in time $2^{\mathcal{O}((d+k)\log{n'})}$ (Corollary~\ref{cor:MPsolution}). 


\begin{remark}
For realizable specifications $\Psi$, transducers are obtained by extracting a winning strategy 
for the system player from the MP game induced by $\mathcal{A}'_\Psi$ (see Remark~\ref{rem:MPstrat}).
\end{remark}

\section{Solution of MP Games}\label{sec:solution}

We show how MP games can be reduced (symbolically) to EL games. Such reductions enable the use of the solution algorithm for EL games from \cite{DBLP:conf/fossacs/HausmannLP24}
to symbolically solve MP games.

Fix a Manna-Pnueli game $G=(A,O)$ with arena $A=(V,E)$ and objective $O=(\Gamma,\gamma,\varphi)$,
and assume w.l.o.g. that each atom in $\varphi$ uses a unique event.
Let $G$ have $n=|V|$ nodes, $m$ edges, $k$ EL events, and $d$ local events.

In the first step, equip the arena $A$ with additional memory
to store the local events that have been visited (or violated) in a play so far. 
If $A$ already has memory for the local events, then this step can be skipped. This is the
case in our synthesis algorithm when using the automata construction given in Section~\ref{sec:ELsynthesis}, Step 1.:
 by Lemma~\ref{lem:localMemory},
each state $q$ in $D_\Psi$ comes with memory $L_q$ for the local events.

An immediate, but costly, reduction from MP games to EL games
then is to simply treat all local events as Emerson-Lei events. 
The synthesis algorithm
described in Section~\ref{sec:ELsynthesis} above is based on this reduction.

This first reduction can be improved by instead reducing MP games to directed acyclic graphs (DAGs) of EL games with simplified objectives.
The simplified EL objectives are obtained by partially evaluating the original MP objective according to the local events from the auxiliary memory.
Then the reduced DAG of EL games can be solved by solving the individual EL games in a bottom-up fashion. 
The alternative synthesis algorithm given in Section~\ref{sec:synthesis} leverages this improved reduction.

\smallskip
\noindent\textbf{Arena transformation.}\label{sec:arenaTransformation}
In the transformed arena $A'$, the nodes from the MP game $G$ are annotated
with sets of local events, acting as memory values that keep track of the $\mathsf{F}$-events that have occurred so far 
and the $\mathsf{G}$-events that have occurred in every game step so far.
The moves in $A'$ are the same as in $A$, but update the memory
whenever an $\mathsf{F}$-event occurs for the first time, or a $\mathsf{G}$-event does not occur for the first time.

Let $\Gamma_\mathsf{F}$ and $\Gamma_\mathsf{G}$ denote the sets of $\mathsf{F}$- and $\mathsf{G}$-events in $\varphi$, respectively and 
put $\Gamma_{\mathsf{F},\mathsf{G}}=\Gamma_\mathsf{F}\cup \Gamma_\mathsf{G}$.
 Let $\Gamma_{\mathsf{EL}}$ denote the set of Emerson-Lei events in $\varphi$.

We define a memory update
function $\mathsf{upd}$ that takes as input a game node $v\in V$ and a current memory
value $L\subseteq \Gamma_{\mathsf{F},\mathsf{G}}$,
and computes the set $\mathsf{upd}(v,L)$ obtained from $L$ by removing all
$\mathsf{G}$-events that do not occur at $v$, and adding all 
$\mathsf{F}$-events that occur at $v$.
Formally, we put
\begin{align*}
	\mathsf{upd}(v,L)=(L\cap \gamma(v))\cap \Gamma_\mathsf{G}\cup(L\cup\gamma(v))\cap \Gamma_\mathsf{F}.
\end{align*}

Then the reduced game arena $A'=(V',E')$ incorporating the auxiliary memory is defined by putting
$V'=V\times 2^{\Gamma_\mathsf{F,G}}$ and
$E'(v,L)=E(v)\times\{\mathsf{upd}(v,L)\}$;
nodes $(v,L)\in V'$ are owned by the owner of $v$ in $G$.

The memory update is defined in such a way that (not) visiting a local event once in a play on $A$ corresponds
to eventually forever (not) having the event in the auxiliary memory in a play on $A'$.

We also define a labelling function $\gamma'(v,L)=(\gamma(v)\cap\Gamma_{\mathsf{EL}})\cup\{L\}$, marking
nodes $(v,L)\in V'$ with the union of the Emerson-Lei events of $v$ and the local events from the
auxiliary memory $L$.
\vspace{-10pt}

\paragraph{Reduction to EL Games.}\label{sec:ELreduction1}
Define the EL objective $O_1=(\Gamma,\gamma',\varphi_1)$, where
the EL formula $\varphi_1$ is obtained from $\varphi$ by replacing all atoms of
shape $\mathsf{F}\,a$ or $\mathsf{G}\,a$ in $\varphi$ with $\mathsf{GF}\,a$.
The objective $O_1$ consists of plays for which the EL events visited by the corresponding play in the
original game together with the events from the auxiliary memory satisfy the adapted formula.
We obtain an EL game $G_1=(A',O_1)$ of size $\mathcal{O}(n\cdot 2^d)$ that has $k+d$ EL events.

\begin{lemma}\label{lemma:MPtoELcorrectness}
    The games $G$ and $G_1$ are equivalent.
\end{lemma}	

\begin{proof}[Sketch]
The claim follows by showing that (not) visiting a local event \emph{once} in $G$ corresponds
to eventually (not) having this event in the auxiliary memory \emph{forever} in $G_1$ (similar to Lemma~\ref{lem:localMemory}). Thus atoms
$\mathsf{F}\, a$ and $\mathsf{G}\, a$ in $G$ correspond to atoms $\mathsf{GF}\,a$ in $G_1$. In fact,
we could also choose $\mathsf{FG}\,a$ to represent local atoms in $G_1$.
\end{proof}

Using Theorem~\ref{thm:ELsolution}, we thus obtain

\begin{corollary}\label{cor:MPsolution}
Manna-Pnueli games with $n$ nodes, $k$ Emerson-Lei events and $d$ local events can be solved in time $\mathcal{O}((k+d)!\cdot (n\cdot 2^d)^{k+d+2})\in 2^{\mathcal{O}(d(k+d)\log{n})}$;
winning strategies require at most $(k+d)!$ memory values.
\end{corollary}

The above bound on solution time improves to
$\mathcal{O}((k+d)!\cdot n^{k+d+2})\in 2^{\mathcal{O}((k+d)\log{n})}$
for games over arenas that have memory for local events (such as the arenas
constructed in Sections~\ref{sec:ELsynthesis} and~\ref{sec:synthesis}). In such cases, 
the arena transformation is not required so that $G_1$ has just $n$ nodes and $k+d$ EL events.
\vspace{-20pt}

\paragraph{Reduction to Compositions of EL Games.}\label{sec:ELreduction2}

Whenever the memory value $L$ changes by a move in the arena $A'$, at least one
$\mathsf{G}$-event is permanently removed from $L$, or at least one $\mathsf{F}$-event is permanently added to $L$. As there are finitely many events, the memory changes finitely often. Consequently, the memory values partition the arena $A'$ into $2^{|\Gamma_\mathsf{F,G}|}$ subarenas; 
we denote the subarena with the memory values fixed to $L$ by $A'_L$.
Together with the memory changing moves, the partition of
$A'$ into the arenas $A'_L$ forms a DAG with top and bottom subarenas $A'_{\Gamma_{\mathsf{G}}}$ and $A'_{\Gamma_{\mathsf{F}}}$, respectively. 

Given an MP formula $\psi$ and a set $L\subseteq \Gamma_{\mathsf{F},\mathsf{G}}$ of local events, 
we let $\psi_{L}$ denote the formula that is obtained
from $\psi$ by evaluating $\mathsf{F}$- and $\mathsf{G}$-atoms according
to $L$, that is by replacing
atoms $\mathsf{F}\, a$ and $\mathsf{G}\, a$ with $\top$ if $a\in L$, and 
with $\bot$ otherwise.
We point out that $\psi_{L}$ does not contain any $\mathsf{F}$- and $\mathsf{G}$-atoms and hence is an EL formula. 

Given a play $\pi\in\mathsf{plays}(A)$,
we define the set of $\mathsf{F}$- and
$\mathsf{G}$-events that are satisfied by $\pi$ by
\begin{align*}
\Gamma_\mathsf{F,G}(\pi)&=\{c\in \Gamma_{\mathsf{F}}\mid
\pi\models \mathsf{F}\, c\} \cup \{c\in \Gamma_{\mathsf{G}}\mid
\pi\models \mathsf{G}\, c\}.
\end{align*}

The following lemma relates MP objectives to
(simplified) EL objectives; the proof is immediate.
\begin{lemma}\label{lem:restr}
For all plays $\pi\in\mathsf{plays}(A)$ we have
$\pi\models \varphi$ if and only if 
$\pi\models\varphi_{\Gamma_\mathsf{F,G}(\pi)}$.
\end{lemma}


Based on Lemma~\ref{lem:restr}, we simplify the objective formula $\varphi_1$ from $G_1$ by using,
for each subarena $A'_L$, the simpler objective $\varphi_L$. 
Formally, we define
$\varphi_2=\textstyle\bigvee_{L\in 2^{\Gamma_\mathsf{F,G}}}(\textstyle\bigwedge_{a\in L} \mathsf{Inf} \,a\land \varphi_{L})$.
This formula expresses the existence of some memory value $L$ such that the auxiliary memory eventually stabilizes to $L$ and
such that $\varphi_{L}$ is satisfied. Plays satisfying this formula eventually stay within some subarena $A'_L$ forever and 
satisfy $\varphi_{L}$.
We put $O_2=(\Gamma,\gamma',\varphi_2)$
and let $G_2$ denote the EL game $(A',O_2)$.

\begin{example}
We demonstrate the reduction of MP games to 
DAGs of EL games for the game from Example~\ref{ex:MP} (not depicting labelings with events for readability).
The objectives of the individual subgames are as follows.
\begin{align*}
\varphi_{\emptyset}&=\varphi[\mathsf{G}\, d\mapsto \bot,\mathsf{F}\, c\mapsto \bot]=\bot\\
\varphi_{\{c\}}&=\varphi[\mathsf{G}\, d\mapsto \bot,\mathsf{F}\, c\mapsto \top]=\mathsf{FG}\,\neg b\\
\varphi_{\{d\}}&=\varphi[\mathsf{G}\, d\mapsto \top,\mathsf{F}\, c\mapsto \bot]=\mathsf{GF}\,a\land\mathsf{GF}\,b\\
\varphi_{\{c,d\}}&=\varphi[\mathsf{G}\, d\mapsto \top,\mathsf{F}\, c\mapsto \top]=(\mathsf{GF}\,a\land\mathsf{GF}\,b)\lor\mathsf{FG}\,\neg b
\end{align*}
    
    \begin{center}
\scalebox{0.6}{
\tikzset{every state/.style={minimum size=15pt}, every node/.style={minimum size=15pt}}
    \begin{tikzpicture}[
		auto,
    node distance=1.5cm,
    thick
    ]
    \tikzstyle{every state}=[
        draw = black,
        thick,
        fill = white,
        minimum size = 4mm
    ]
    
     \node[state, label={}] (1a)  {};
     \node[state, rectangle,label={}] [left of=1a] (2a) {};
     \node[state,label={}] [below of=1a] (3a) {};
     \node[state, rectangle,label={}] [right of=1] (4a) {};
     \node[state,label={}] [below of=4a] (5a)  {};
     \node[state,label={}] [right of=5a] (6a)  {};

     \node[label={$G_{\{d\}}$}] (0a) at (-2.7,0) {};

     \path[->] (1a) edge node {} (2a);

     \path[->,dotted] (1a) edge[bend left=15] node {} (3a);
     \path[->] (3a) edge[bend left=15] node {} (1a);

     \path[->,dotted] (1a) edge[bend left=15] node {} (4a);
     \path[->] (4a) edge[bend left=15] node {} (1a);

     \path[->] (4a) edge node {} (5a);

     \node[state, label={}] (1b) at (7,0) {};
     \node[state, rectangle,label={}] [left of=1b] (2b) {};
     \node[state,label={}] [below of=1b] (3b) {};
     \node[state, rectangle,label={}] [right of=1b] (4b) {};
     \node[state,label={}] [below of=4b] (5b)  {};
     \node[state,label={}] [right of=5b] (6b)  {};

     \node[label={$G_{\emptyset}$}] (0b) at (4.3,0) {};

     \path[->,dotted] (1b) edge node {} (2b);

     \path[->] (1b) edge[bend left=15] node {} (3b);
     \path[->] (3b) edge[bend left=15] node {} (1b);

     \path[->] (1b) edge[bend left=15] node {} (4b);
     \path[->] (4b) edge[bend left=15] node {} (1b);

     \path[->] (4b) edge[bend left=15] node {} (5b);
     \path[->] (5b) edge[bend left=15] node {} (4b);

     \path[->,dotted] (5b) edge node {} (6b);

     \node[state, label={}] (1c) at (0,-3.5) {};
     \node[state, rectangle,label={}] [left of=1c] (2c) {};
     \node[state,label={}] [below of=1c] (3c) {};
     \node[state, rectangle,label={}] [right of=1c] (4c) {};
     \node[state,label={}] [below of=4c] (5c)  {};
     \node[state,label={}] [right of=5c] (6c)  {};

     \node[label={$G_{\{c,d\}}$}] (0c) at (-2.7,-3.5) {};

     \path[->] (1c) edge[bend left=15] node {} (2c);
     \path[->] (2c) edge[bend left=15] node {} (1c);

     \path[->,dotted] (1c) edge[bend left=15] node {} (3c);
     \path[->] (3c) edge[bend left=15] node {} (1c);

     \path[->] (1c) edge[bend left=15] node {} (4c);
     \path[->] (4c) edge[bend left=15] node {} (1c);

     \path[->] (4c) edge node {} (5c);
     
     \path[->] (6c) edge node {} (5c);

     \path[->] (2c) edge[bend left=30] node {} (4c);

     \node[state, label={}] (1d) at (7,-3.5) {};
     \node[state, rectangle,label={}] [left of=1d] (2d) {};
     \node[state,label={}] [below of=1d] (3d) {};
     \node[state, rectangle,label={}] [right of=1d] (4d) {};
     \node[state,label={}] [below of=4d] (5d)  {};
     \node[state,label={}] [right of=5d] (6d)  {};

     \node[label={$G_{\{c\}}$}] (0d) at (4.3,-3.5) {};

     \path[->] (1d) edge[bend left=15] node {} (2d);
     \path[->] (2d) edge[bend left=15] node {} (1d);

     \path[->,dotted] (1d) edge[bend left=15] node {} (3d);
     \path[->] (3d) edge[bend left=15] node {} (1d);

     \path[->] (1d) edge[bend left=15] node {} (4d);
     \path[->] (4d) edge[bend left=15] node {} (1d);

     \path[->] (4d) edge[bend left=15] node {} (5d);
     \path[->] (5d) edge[bend left=15] node {} (4d);

     \path[->,dotted] (5d) edge[bend left=15] node {} (6d);
     \path[->] (6d) edge[bend left=15] node {} (5d);

     \path[->] (2d) edge[bend left=30] node {} (4d);

     \path[->,dashed] (2a) edge node[right, pos=0.715] {$c$} (2c);
     \path[->,dashed] (6a) edge node[right, pos=0.23] {$c$} (6c);
     \path[->,dashed] (5a) edge[bend right=15] node[above, pos=0.35] {$\neg d$} (5b);
     \path[->,dashed] (2b) edge node[right, pos=0.715] {$c$} (2d);
     \path[->,dashed] (6b) edge node[right, pos=0.23] {$c$} (6d);
     \path[->,dashed] (5c) edge[bend right=15] node[above, pos=0.35] {$\neg d$} (5d);

  \begin{pgfonlayer}{background}
    \filldraw [line width=4mm,join=round,black!10]
      (-1.9,-2.0)  rectangle (3.3,0.7);
  \end{pgfonlayer}

  \begin{pgfonlayer}{background}
    \filldraw [line width=4mm,join=round,black!10]
      (5.1,-2.0)  rectangle (10.3,0.7);
  \end{pgfonlayer}

  \begin{pgfonlayer}{background}
    \filldraw [line width=4mm,join=round,black!10]
      (-1.9,-5.5)  rectangle (3.3,-2.8);
  \end{pgfonlayer}

  \begin{pgfonlayer}{background}
    \filldraw [line width=4mm,join=round,black!10]
      (5.1,-5.5)  rectangle (10.3,-2.8);
  \end{pgfonlayer}

  \end{tikzpicture}
  }

\end{center}

Dashed edges depict changes in the memory for local events (corresponding to first-time visits of event $c$ or $\neg d$, respectively);
they descend in the DAG. Dotted edges depict winning strategies for each subgame. An overall winning strategy is obtained
by using additional memory to keep track of the current subgame and by always moving according to the winning strategy for the current subgame.
\end{example}

\begin{lemma}\label{lemma:MP-to-EL-correctness}
    The games $G$ and $G_2$ are equivalent.
\end{lemma}	
\begin{proof}[Sketch]
For one direction, use a winning strategy for $G$ 
to play in $G_2$, simply ignoring memory values. By construction, 
the resulting plays in $G_2$ eventually stay within one subgame
$G_L$ and are winning by Lemma~\ref{lem:restr}.
For the converse direction, a winning strategy for
$G_2$ provides, for each set $L$ of 
local events, a strategy to play in the subgame $G_L$.
In $G$, let system player always follow the strategy for 
the subgame that corresponds to the local events visited so far.
The resulting plays in $G$ are again winning by Lemma~\ref{lem:restr}. 
\end{proof}

\begin{remark}[Strategy extraction]\label{rem:MPstrat}
For each EL subgame $G_L$, construct a winning strategy $\sigma_L$ (according to Remark~\ref{rem:ELstrat}).
An overall strategy $\sigma$ for $G$ uses auxiliary memory $2^{\Gamma_{\mathsf{F},\mathsf{G}}}$ to keep track of the local events that have
been satisfied / violated so far. This memory identifies, at each point, a subgame $G_L$. 
Define $\sigma$ to always play according to $\sigma_L$, where $L$ is
the current content of the auxiliary memory.
\end{remark}

Due to its particular DAG structure, the game $G_2$ can be solved by solving the subgames $G_L$ (played over arena $A'_L$ with objective $\varphi_L$) individually, starting from the bottom game $G_{\Gamma_\mathsf{F}}$. Once a subgame $G_L$ has been solved, all edges in the remaining subgames that lead to a node $v\in G_L$ are marked as winning or losing, depending on whether
$v$ is winning or losing in $G_L$. Then $G_L$ is removed from the DAG and one of the remaining bottom subgames is picked and solved in turn.
Thus $G_2$ can be solved by solving at most $2^d$ EL games,
each of which has at most $n$ nodes, $m$ edges, and an objective with at most $k$ EL events.
The overall time complexity of solving $G_2$ in this way is in $\mathcal{O}(2^d\cdot m\cdot k!\cdot n^k)$. Note that we have $m\leq n^2$.

\begin{theorem}\label{thm:MPDAGsolution}
Manna-Pnueli games with $n$ nodes, $k$ Emerson-Lei events and $d$ local events can be solved in time $\mathcal{O}(k!\cdot n^{k+2}\cdot 2^d)\in 2^{\mathcal{O}(d+k\log{n})}$;
winning strategies require at most $k!\cdot 2^d$ memory values.
\end{theorem}

The above bounds improve to a solution time $\mathcal{O}(m\cdot k!\cdot n^k)\in 2^{\mathcal{O}(k\log n)}$ and
strategy size $k!$ for games over arenas that have memory for local events. Thus the solution complexity for MP games with memory for local events
is the same as for EL games.

The dependency on the number of events in Theorem~\ref{thm:ELsolution}~(and in Corollary~\ref{cor:MPsolution}) is
factorial while the dependency on the number of local events in Theorem~\ref{thm:MPDAGsolution} is exponential.

\begin{figure}[t]
\begin{footnotesize}
    \centering
\begin{tikzpicture}[
    show background rectangle]

    \node[module] (I1) {\textrm{\LTLfp/\PPLTLp formula} $\Psi$};
    \node[module, below=4mm of I1] (I2) {\textrm{Boolean formula over} $\mathbb{Q}_i\Phi_i$};
    \node[below=6.5mm of I2] (yo) {};
    \node[module, left= 8mm of yo] (I4) {\textrm{EL game} $\mathcal{A}_\psi$};
    \node[module, right= 8mm of yo] (I3) {\textrm{MP game} $\mathcal{A}'_\psi$};
    \node[module, below=4mm of I3] (I5) {\textrm{DAG of EL games}};
    \node[module, below=28mm of I2] (I7) {\textrm{output: ``(un)realizable'', strategy}};

    \path[->] (I1) edge node[right, pos=0.55] {parser} (I2);
    \path[->] (I2) edge node[right, pos=0.45] {\,\,\,\,Sect.~\ref{sec:synthesis}} (I3);
    \path[->] (I2) edge node[left, pos=0.45] {Sect.~\ref{sec:ELsynthesis}\,\,\,} (I4);
    \path[->] (I3) edge node[right, pos=0.4] {\,Sect.~\ref{sec:ELreduction2}} (I5);
    \path[->] (I4) edge node[left, pos=0.5] {Thm.~\ref{thm:ELsolution}} (I7);
    \path[->] (I5) edge node[right, pos=0.5] {\,\,Thm.~\ref{thm:MPDAGsolution}} (I7);
    \path[->, dashed] (I5) edge[bend left=10] node[below] {} (I4);
    \path[->, dashed] (I4) edge[bend left=10] node[below, pos=0.3] {Thm.~\ref{thm:ELsolution}\,} (I5);



    \path (I1.west)--node[above=4mm, font=\sffamily\bfseries] (arc) {\tool} (I1.west);

\end{tikzpicture}
    \caption{Overall structure of the implementation}
    \label{fig:impl}
\end{footnotesize}
\end{figure}
\section{Implementation}\label{sec:impl}
We implemented a prototype \tool~\cite{lydiasyftplus}~(see Figure~\ref{fig:impl}) that realizes both our method and the existing approach based on a reduction to Emerson-Lei games~as described in Section~\ref{sec:ELsynthesis}. While \tool is the first implementation of an \LTLfp/\PPLTLp synthesizer, it also enables a direct empirical comparison between the two approaches. 

\smallskip
\noindent\textbf{Synthesis via EL solver.}
We implemented an \LTLfp/\PPLTLp synthesis procedure \tool-EL based on a reduction to EL games, using the symbolic EL solution algorithm. For \LTLfp synthesis, we leverage the existing \LTLf-to-DFA translator \lydiasyft~\cite{lydiasyft} as a backend to translate \LTLf formula components to explicit-state deterministic finite automata~(DFAs). These DFAs are tailored according to the reduction in~\cite{LTLfplus}~(cf. Step 1. in Section~\ref{sec:ELsynthesis}) and transformed into symbolic representation using Binary Decision Diagrams~(BDDs)~\cite{bryant.92.cs,ZTLPV17}. For \PPLTLp synthesis, we implemented a direct symbolic DFA construction based on the method described in~\cite{BonassiGFFGS23icaps}, which also serves as the first direct translator from \PPLTLp to symbolic DFA. Since all DFAs are in symbolic representation, we can avoid explicitly computing their cross-product. Instead, we perform an on-the-fly product construction during the EL game solving, which is carried out using symbolic fixpoint computations as in~\cite{ZTLPV17,DBLP:conf/fossacs/HausmannLP24}.





\smallskip
\noindent\textbf{Synthesis via MP solver.}
We also implemented the MP  based approach to \LTLfp/\PPLTLp synthesis, described in Section~\ref{sec:synthesis}. 
Given an MP game $G =(V,E,(\Gamma,\gamma,\varphi))$, the core component in solving the game involves constructing a directed acyclic graph~(DAG) of EL games. These EL games are then solved in a bottom-up order along the DAG. 
Note that each EL game in the DAG is obtained by evaluating the local events in the MP condition $\varphi$. We encode $\varphi$ symbolically using BDDs for efficient evaluation and simplification.

\section{Experiments}\label{sec:bench}

All experiments were conducted on a laptop running 64-bit Ubuntu 22.04.4 LTS, with an i5-1245U CPU with 12 cores and 32GB of memory. Time-out was set to one hour.



\begin{figure}[t]
    \centering
    \scalebox{0.9}
    {\includegraphics[width=\linewidth]{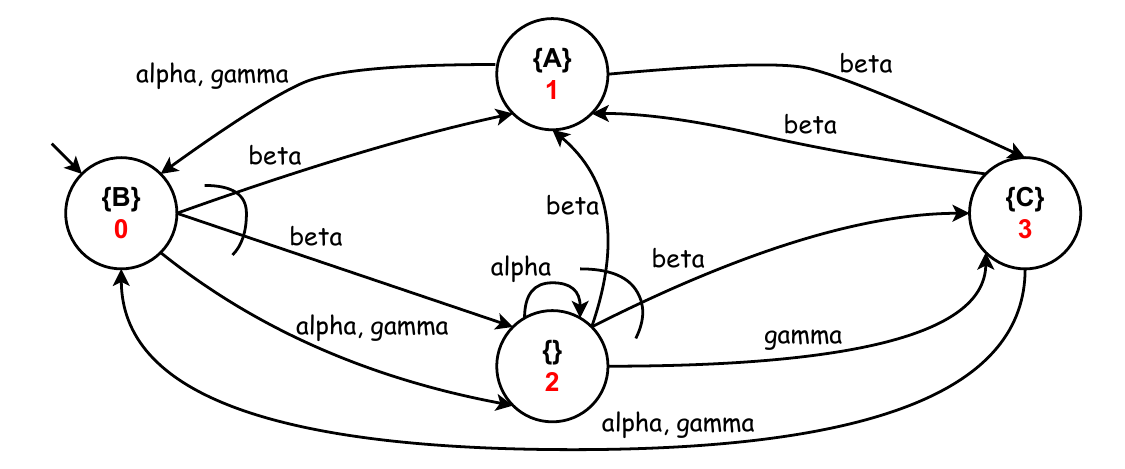}}
    \caption{A simple planning domain}
    \label{fig:domain}
\end{figure}

\smallskip
\noindent\textbf{Comparing \LTLfp and \LTLf.} 
We construct the benchmark example based on a simple planning problem in a nondeterministic domain, illustrated in Figure~\ref{fig:domain}. 
In this domain, a robot starts at state 0 and chooses an action from the set \{\emph{alpha, beta, gamma}\} to move to the next state. However, the actual next state also depends on the response of the environment, which is nondeterministic to the robot. For example, if the robot is in state 0 and takes action \emph{beta}, the robot may move to either state 1 or state 2. We can describe this domain in either \LTLf or \LTLfp. The core is that we just use the formula to specify safety conditions that represent the transitions of the domain following the robot action and the environment response. In this case, we have an \LTLf formula $\Phi_D$ caturing the traces of the domain, and an \LTLf formula $\Phi_{act}$ specifying that exactly one action is executed at each step. 
The objective of the robot is given by the \LTLf formula $\Phi_{goal} = \mathsf{F}(A \wedge \mathsf{F}(B \wedge \mathsf{X}\,\mathsf{false}))$, specifying that the robot must eventually visit $A$ and then visit $B$ at the end of the finite trace; recall that $\mathsf{X}$ stands for ``weak next''. 
Ultimately, we have the overall \LTLf specification describing the planning problem as: 
$\Phi = \Phi_{act} \wedge (\Phi_D \rightarrow \Phi_{goal}).$

To construct the corresponding \LTLfp specification for the same domain, we note that both $\Phi_{act}$ and $\Phi_D$ represent safety conditions. This allows us to simply add universal quantification $\forall$ in front of $\Phi_{act}$ and $\Phi_D$. For the objective, instead of simply requiring the robot to visit $A$ and then (or meanwhile) $B$ once, we make it more difficult by requiring that the robot should ``visit $A$ and subsequently (or simultaneously) visit $B$" infinitely often. This leads to the \LTLfp specification 
$\Psi = \forall \Phi_{act} \wedge (\forall \Phi_D \rightarrow \forall \exists \Phi_{goal}).$

We run both problems using LydiaSyft for \LTLf synthesis on $\Phi$ and \tool~(both using the MP solver and the EL solver) for \LTLfp synthesis on $\Psi$, and observe that both solvers complete within a few milliseconds. Recall that $\Psi$ encodes a more complex objective involving a recurrence formula. This confirms that there is no overhead in constructing the arena for \LTLfp respect to \LTLf, and that the more complex fixpoint computation required for \LTLfp does not compromise performance.

\begin{figure}[t]
    \centering
    \includegraphics[width=\linewidth]{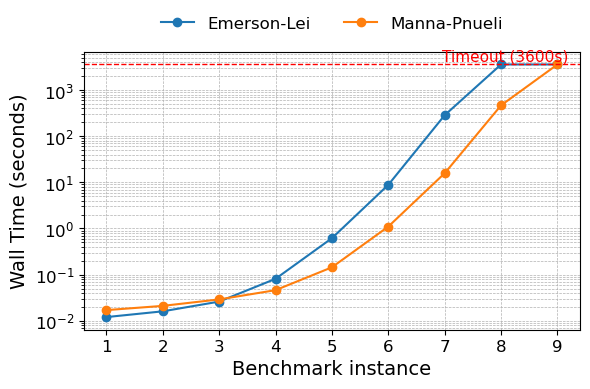}
    \caption{Runtime of \LTLfp synthesis, counter-games}
    \label{fig:ltlfp-pattern3-time}
\end{figure}

\smallskip
\noindent\textbf{Comparing MP solver and the EL solver.}
We next compare the scalability of the MP solver and the EL solver. 
%
The first benchmark  is inspired by the counter-game introduced in~\cite{ZhuGPV20}. We adapt the basic game setting which involves an $n$-bit binary counter. We use $\forall \Phi_B$ to describe the behaviour of the counter, including both the value transitions and the response of the environment to grant requests to increment the counter. The formula $\forall \exists \Phi_{add}$ represents agent requests to increase the counter value, while $\exists \Phi_g$ captures the goal condition that the counter eventually reaches its maximum value, i.e., all bits are set to~1. The \LTLfp formulas in this series then are defined as 
\begin{align*}
&\Psi = \forall (\Phi_{init} \wedge \Phi_{inc} \wedge \Phi_{B_i}) \rightarrow (\forall\exists \Phi_{add} \rightarrow \exists \Phi_g), \mbox{ where} \\
&\Phi_{init} = \neg c_0 \wedge \ldots \wedge \neg c_{n-1} \wedge \neg b_0 \wedge  \ldots \neg b_{n-1}\\
&\Phi_{inc} = \mathsf{G}(add \to (\mathsf{X}(c_0) \wedge \mathsf{X}\mathsf{X}(c_0) \wedge \mathsf{X}\mathsf{X}\mathsf{X}(c_0))\\
&\Phi_{B_i} = \begin{cases}
 ((\neg c_i \wedge \neg b_i) \rightarrow \mathsf{X}(\neg b_i \wedge \neg c_{i+1}))\,\land\\
 ((\neg c_i \wedge b_i) \rightarrow \mathsf{X}(b_i \wedge \neg c_{i+1}))\,\land\\
 ((c_i \wedge \neg b_i) \rightarrow \mathsf{X}(b_i \wedge \neg c_{i+1}))\,\land \\
 ((c_i \wedge b_i) \rightarrow \mathsf{X}(\neg b_i \wedge c_{i+1})) \end{cases} \\
& \Phi_{add} = \mathsf{F}(add \wedge \mathsf{X}(\mathsf{false}))\\
& \Phi_g = \mathsf{F}(b_0 \wedge \ldots \wedge b_{n-1} \wedge \mathsf{X}(\mathsf{false}))
\end{align*}

In this benchmark, all variables $b_i$ and $c_i$ are environment variables, since they are used to describe the status of the counter. The variable $add$ is an agent variable, used to request an increment of the counter. All the formulas in this benchmark  are realizable as well. Figure~\ref{fig:ltlfp-pattern3-time} shows that the MP solver significantly outperforms the EL solver, solving more instances with lower runtime.

\begin{figure}[t]
    \centering
    \includegraphics[width=\linewidth]{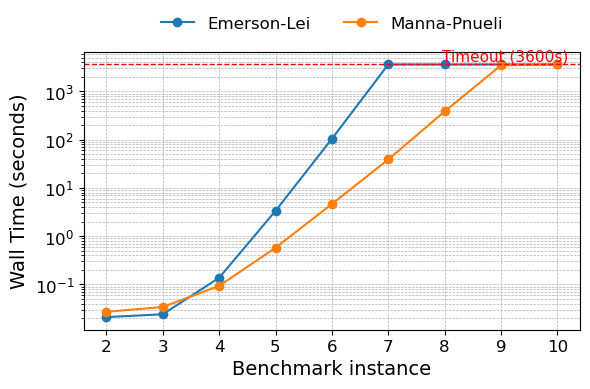}
    \caption{Runtime of \LTLfp synthesis, \emph{$(\forall\exists \rightarrow \exists)$-Pattern}}
    \label{fig:ltlfp-pattern2-time}
\end{figure}

\smallskip
To check scaling as the game conditions increase, we construct a benchmark consisting of formulas formed as conjunctions of subformulas in the form of $\forall\exists \rightarrow \exists$, referred to as \emph{$(\forall\exists \rightarrow \exists)$-Pattern}, and is defined as:
\begin{align*}
    \Psi_n &= \textstyle\bigwedge_{1 \leq i \leq n} (\,\forall \exists\, \mathsf{F}(e_i \land \mathsf{X}(\mathsf{false})) \rightarrow \exists\, \mathsf{F}(a_i \land \mathsf{X}(\mathsf{false}))\,)
\end{align*}

Each \LTLf formula component is very simple. But the number of recurrence and guarantee formulas increases with $n$.
Each $a_i$ is an agent variable, while each $e_i$ is an environment variable; all formulas in this series are realizable.

Figure~\ref{fig:ltlfp-pattern2-time} shows the running time on the $(\forall\exists \rightarrow \exists)$-Pattern benchmarks. The MP solver is able to solve more instances within the time limit, and consistently takes less time than the EL solver, demonstrating the superior performance of the MP solver.

\begin{figure}[t]
    \centering
    \includegraphics[width=\linewidth]{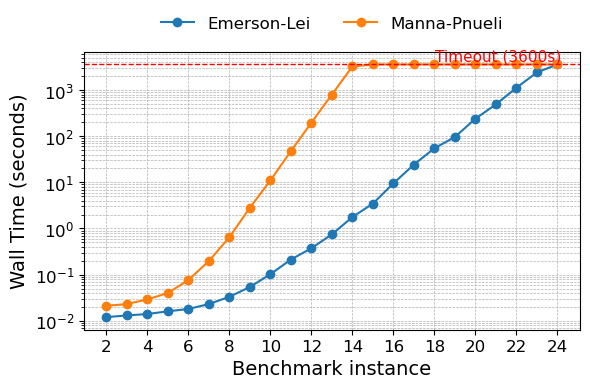}
    \caption{Runtime of \LTLfp synthesis, \emph{$\exists$-Pattern}}
    \label{fig:ltlfp-pattern7-time}
\end{figure}

\smallskip
A natural question is whether there are cases for which the EL solver performs better than the MP solver in practice.
We construct a specific benchmark, the instances of which contain only guarantee formulas~($\exists$). More specifically, the formulas, referred to as \emph{$\exists$-Pattern}, are defined as follows:
\begin{align*}
\Psi_n = \textstyle\bigwedge_{1 \leq i \leq n}\exists\, \mathsf{F} ((a_i \lor e_i) \land \mathsf{X}(\mathsf{false}))
\end{align*}

Since instances contain only local events, one might expect the MP solver to outperform the EL solver. However, the results in Figure~\ref{fig:ltlfp-pattern7-time} show that the EL solver is able to solve a substantially larger number of instances within the time limit and consistently takes less time than the MP solver on instances that both solvers can handle. To understand why this happens, one has to consider that while the EL solver computes a nested fixpoint, having to reduce $\exists(\cdot)$ to $\forall\exists(\cdot)$, the automata construction for $\exists(\cdot)$ introduces loops on the accepting states, thus shortcutting the corresponding nested fixpoint. On the other hand, the MP solver must first construct a DAG, where each node corresponds to a relatively~(but in this case not significantly) simpler EL game. Moreover, the MP solver can only conclude ``realizable'' when all these games have been solved, resulting in greater overall computational overhead. 








\section{Discussion}

We have implemented synthesizers for \LTLfp (and \PPLTLp) based on reductions to EL and MP games. We have introduced MP games and showed that theoretically 
solvers based on MP games have better computational characteristics than EL solvers.
In practice, experimental results indicate that while the MP solver often performs better than the EL solver, this is not always the case, since there is a trade-off between efficiency gains obtained by solving
DAGs of simpler EL games versus the higher cost associated to the construction of these DAGs.
This indicates that fine-tuning may be needed to strike the perfect balance between which local events should be treated
by the DAG construction and which should be left to the EL solver. We leave this for future work.

We also plan to explore a recently proposed alternative game construction that 
combines the benefits of the EL and the MP reductions, eliding the DAG construction
by integrating local events into the existing EL events, instead of adding auxiliary EL events \cite{Duret-Lutz25}.

We conclude by briefly mentioning \emph{obligation properties},
that is, Boolean combinations of guarantee, $\exists(\cdot)$, and safety, $\forall(\cdot)$ formulas.
Synthesis for these formulas reduces to the solution of MP games without EL events; asymptotically, the solution
of such games is not harder than the solution of reachability games, indicating that particularly good performance can be obtained for the corresponding fragments of \LTLfp (and \PPLTLp). Again, we leave this for future work.

\section*{Acknowledgments}
This work is supported in part by the ERC Advanced Grant WhiteMech (No. 834228), the PRIN project RIPER (No. 20203FFYLK), the PNRR MUR project FAIR (No. PE0000013), the UKRI Erlangen AI Hub on Mathematical and Computational Foundations of AI, the EPSRC project LaST (No. EP/Z003121/1), and by the Swedish research council (VR) through project No. 2020-04963. Gianmarco Parretti is supported by the Italian National Phd on AI at the University of Rome ``La Sapienza''.

\clearpage
\bibliographystyle{kr}
\bibliography{lib,gdg-references}


\clearpage
\section*{Supplement to ``Preliminaries''}

\paragraph{Syntax and Semantics of \LTL.}
Linear-time Temporal Logic (\LTL)~\cite{Pnueli77} allows to express temporal properties of infinite traces.
The set of \LTL formulas over the set $AP$ of atomic propositions is 
given by the following grammar.
\begin{align*}
\varphi,\psi::= p \mid \neg\varphi\mid \varphi\land\psi\mid \mathsf{X} \varphi \mid \varphi\mathsf{U}\psi \tag{$p\in AP$}
\end{align*}
We make use of common abbreviations such as $\varphi\lor\psi=\neg(\neg \varphi\land\neg\psi)$, $\mathsf{true}=p\lor \neg p$, $\mathsf{false}=\neg\mathsf{true}$, $\mathsf{F}\varphi=\mathsf{true}\mathsf{U}\varphi$ (``eventually'') and
$\mathsf{G}\varphi=\neg \mathsf{F}\neg\varphi$ (``always'').
\LTL formulas are evaluated over infinite traces $\tau\in (2^{AP})^\omega$ of sets of atomic propositions. 
Satisfaction of \LTL formulas by infinite traces is defined inductively as follows, referring to the $i$th element in $\tau$ by $\tau_i$.
\begin{align*}
\tau,i&\models p &\text{ iff } &p\in \tau_i\\
\tau,i&\models \neg\varphi &\text{ iff } &\tau,i\not\models\varphi\\
\tau,i&\models \varphi\land\psi &\text{ iff } &\tau,i\models \varphi \text{ and } \tau,i\models \psi\\
\tau,i&\models \mathsf{X}\varphi &\text{ iff } &\tau,i+1\models\varphi\\
\tau,i&\models \varphi\mathsf{U}\psi &\text{ iff } &\exists j\geq i.\,\tau,j\models \psi \text{ and } \\
&&&\forall i\leq j'<j.\tau,j'\models \varphi.
\end{align*}
Given an \LTL formula $\varphi$, we let 
\[
[\varphi]=\{\tau\in (2^{AP})^\omega\mid \tau,0\models\varphi\}
\]
denote
the set of infinite traces that satisfy $\varphi$ at the start.

\paragraph{Syntax and Semantics of \LTLf.}
The syntax of \LTL on finite traces (\LTLf)~\cite{DegVa13} is the same as the syntax of \LTL given above. However, \LTLf formulas are evaluated
over finite traces $\tau\in (2^{AP})^*$ rather than over infinite ones. A common abbreviation is the operator $\mathsf{X}\varphi=\neg \mathsf{X}[!]\neg\varphi$ (``weak next''),
expressing that if there is a next position in the trace, then it satisfies $\varphi$.
The satisfaction of \LTLf formulas by finite traces is defined inductively
as follows, where $|\tau|$ denotes the length of a finite trace $\tau$.
\begin{align*}
\tau,i&\models p &\text{ iff } &p\in \tau_i\\
\tau,i&\models \neg\varphi &\text{ iff } &\tau,i\not\models\varphi\\
\tau,i&\models \varphi\land\psi &\text{ iff } &\tau,i\models \varphi \text{ and } \tau,i\models \psi\\
\tau,i&\models \mathsf{X[!]}\varphi &\text{ iff } &i+1< |\tau|\text{ and }\tau,i+1\models\varphi\\
\tau,i&\models \varphi\mathsf{U}\psi &\text{ iff } &\exists i\leq j<|\tau|.\,\tau,j\models \psi \text{ and } \\
&&&\forall i\leq j'<j.\tau,j'\models \varphi.
\end{align*}
The subtle (but consequential) difference to standard \LTL semantics is the requirement that $i+1<|\tau|$ (resp. $j<|\tau|$) in the last two clauses;
that is, for all $\varphi$, $\mathsf{X[!]}\varphi$ is not satisfied at the end of a finite trace, and in
order for $\varphi\mathsf{U}\psi$ to be satisfied in a finite trace, $\psi$ is required to be satisfied before the trace ends.
The formula $\mathsf{last} = \neg \mathsf{X[!]}\,\mathsf{true}$ is satisfied exactly at the last position of a finite trace.

Given an \LTLf formula $\varphi$, we let 
\[
[\varphi]=\{\tau\in (2^{AP})^*\mid \tau,0\models\varphi\}
\]
denote the set of \emph{finite} traces that satisfy $\varphi$ at the start.

\paragraph{Syntax and Semantics of \PPLTL.}

Pure past LTL (\PPLTL)~\cite{DeGiacomoSFR20} allows to express temporal properties of finite traces
by means of statements that refer to the past.
The set of \PPLTL formulas over the set $AP$ of atomic propositions is 
given by the following grammar.
\begin{align*}
\varphi,\psi::= p \mid \neg\varphi\mid \varphi\land\psi\mid \mathsf{Y[!]} \varphi \mid \varphi\mathsf{S}\psi \tag{$p\in AP$}
\end{align*}
Here, $\mathsf{Y[!]}$ (``yesterday'') and $\mathsf{S}$ (``since'') are the past operators; they
are past analogues of the temporal operators $\mathsf{X[!]}$ and $\mathsf{U}$.
Common abbreviations include $\mathsf{O}\varphi=\mathsf{true}\,\mathsf{S}\,\varphi$ (``at least once in the past'') and
$\mathsf{H}\varphi=\neg \mathsf{O}\neg\varphi$ (``historically''). 
The satisfaction of \PPLTL formulas by finite traces is defined inductively
as follows.
\begin{align*}
\tau,i&\models p &\text{ iff } &p\in \tau_i\\
\tau,i&\models \neg\varphi &\text{ iff } &\tau,i\not\models\varphi\\
\tau,i&\models \varphi\land\psi &\text{ iff } &\tau,i\models \varphi \text{ and } \tau,i\models \psi\\
\tau,i&\models \mathsf{Y[!]}\varphi &\text{ iff } &i>0\text{ and }\tau,i-1\models\varphi\\
\tau,i&\models \varphi\,\mathsf{S}\,\psi &\text{ iff } &\exists i\geq j\geq 0.\,\tau,j\models \psi \text{ and } \\
&&&\forall i\geq j'>j.\tau,j'\models \varphi.
\end{align*}
We observe that for all $\varphi$, $\mathsf{Y[!]}\varphi$ is not satisfied at the start of a finite trace, and in
order for $\varphi\,\mathsf{S}\,\psi$ to be satisfied at position $i$ in a finite trace, $\psi$ has to be satisfied somewhere in the first $i+1$ positions
of the trace.
The formula $\mathsf{first} = \mathsf{Y}\,\mathsf{false}$ (where $\mathsf{Y}=\neg \mathsf{Y}\mathsf[!]\neg$ denotes~``weak yesterday") is satisfied exactly at the first position of a finite trace.

Given a \PPLTL formula $\varphi$, we let 
\[
[\varphi]=\{\tau\in (2^{AP})^*\mid \tau,|\tau|-1\models\varphi\}
\]
denote the set of \emph{finite} traces that satisfy $\varphi$ at the \emph{end}, that is, at position $|\tau|-1$.

\subsection*{Supplement to ``Synthesis via EL Games''}

\textbf{Additional details on Zielonka trees.}

We recall the definition of Zielonka trees from~\cite{DBLP:conf/lics/DziembowskiJW97,DBLP:conf/fossacs/HausmannLP24}.
	
    The \emph{Zielonka tree} for an Emerson-Lei formula $\varphi$ over
	set $\Gamma$ of events is a tuple
	$\mathcal{Z}_\varphi=(T,R,l)$ where $(T,R\subseteq T\times T)$ is a tree and
	$l:T\to 2^\Gamma$ is a labeling function that
	assigns sets $l(t)$ of events to vertices $t\in T$.
	We denote the root of $(T,R)$ by $r$.
    Then $\mathcal{Z}_\varphi$ is defined to be
	the unique tree (up to reordering of child vertices) that satisfies
	the following constraints.
	\begin{itemize}
		\item The root vertex is labeled with $\Gamma$, that is, $l(r)=\Gamma$.
		\item Each vertex $t$ has exactly one child vertex $t_\Delta$
		(labeled with $l(t_\Delta)=\Delta$) for
		each set $\Delta$ of events that is maximal in
$\{\Delta'\subsetneq l(t)\mid \Delta'\models\varphi \Leftrightarrow l(t)\not\models \varphi\}$.
	\end{itemize}
    A vertex $t\in T$ is \emph{winning} if $l(t)\models\varphi$, and \emph{losing} otherwise.

\begin{example}\label{example:zielonkaTree}
As an example of a Zielonka tree that has branching at both winning and losing vertices,
consider the tree for the EL objective
$\varphi_{EL}=(\mathsf{FG}~\neg a\vee
\mathsf{GF}~b)\wedge(\mathsf{FG}~\neg a\vee
\mathsf{FG}~\neg d)\land
\mathsf{GF}~c$.
This property can be seen as the conjunction of a
Streett pair $(a,b)$ with two disjunctive Rabin pairs $(c,a)$ and
$(c,d)$, altogether stating that $c$ occurs infinitely often and $a$ occurs
finitely often or $b$ occurs infinitely often and $d$ occurs finitely
often.
We depict the induced Zielonka tree, with circles depicting losing vertices and boxes depicting winning vertices.
\begin{center}
\tikzset{every state/.style={minimum size=15pt}}
\begin{tiny}
  \begin{tikzpicture}[
		auto,
    node distance=0.8cm,
    semithick
    ]
     \node[state, label={right: $a,b,c,d$}] (0) {$1$};
     \node (yo) [below of=0] {};
     \node (yo2) [right of=yo] {};
     \node[state, rectangle, label={left: $a,b,c$}] (1) [left of=yo] {$2$};
     \node[state, rectangle, label={left: $b,c,d$}] (2) [right of=yo2]
     {$3$};
     \node (yo1) [below of=1] {};
     \node[state, label={right: $a,c$}] (3) [right of=yo1] {$5$};
     \node[state, label={right: $a,b$}] (4) [left of=yo1] {$4$};
     \node[state, label={below: $b,d$}] (5) [below of=2] {$6$};
     \node[state, rectangle, label={right: $c$}] (6) [below of=3] {$7$};
     \node[state, label={right: $\emptyset$}] (7) [below of=6] {$8$};
     \path[->] (0) edge node [pos=0.3,left] {} (1);
     \path[->] (0) edge node [pos=0.3,right] {} (2);
     \path[->] (1) edge node [pos=0.3,left] {} (3);
     \path[->] (1) edge node [pos=0.3,left] {} (4);
     \path[->] (2) edge node [pos=0.3,left] {} (5);
     \path[->] (3) edge node [pos=0.3,left] {} (6);
     \path[->] (6) edge node [pos=0.3,left] {} (7);

  \end{tikzpicture}
\end{tiny}
\end{center}
\end{example}
The solution algorithm in for EL games~\cite{DBLP:conf/fossacs/HausmannLP24}
extracts, from an input objective formula $\varphi$, a $\mu$-calculus formula (equivalently: a system of fixpoint equations)
from the Zielonka tree for $\varphi$. In this formula, winning vertices are encoded by greatest fixpoints, and losing vertices are encoded by least fixpoints.
The nesting depth of alternating fixpoints in the resulting $\mu$-calculus formula is exactly the height of the input Zielonka tree.
Then an Emerson-Lei game can be solved symbolically by symbolic evaluation of the $\mu$-calculus formula.\medskip
\begin{example}~\cite{DBLP:conf/fossacs/HausmannLP24}\label{example:eqsys}
The fixpoint equation system associated to the Zielonka tree from Example~\ref{example:zielonkaTree} is as follows:
\begin{scriptsize}
    \begin{align*}
X_{1}&=_{\mathsf{LFP}} X_{2}\cup X_{3} \qquad X_{2} =_{\mathsf{GFP}} X_{4}\cap X_{5}\qquad
X_{3}=_{\mathsf{GFP}} X_{6}\\
X_{5} &=_{\mathsf{LFP}} X_{7}\qquad \qquad\,\,\,\,
X_{7} =_{\mathsf{GFP}} X_{8}\\
X_{4} &=_{\mathsf{LFP}} (\neg c\land \neg d
\cap\mathsf{Cpre}(X_4))\cup(c\land \neg
d\cap\mathsf{Cpre}(X_{2}))\cup(d\cap\mathsf{Cpre}(X_{1}))\\
X_{6} &=_{\mathsf{LFP}} (\neg a\land\neg c\cap\mathsf{Cpre}(X_6))\cup
(\neg a\land c\cap\mathsf{Cpre}(X_3))\cup(a\cap\mathsf{Cpre}(X_{1}))\\
X_{8} &=_{\mathsf{LFP}}
(\neg a\land\neg b \land \neg c\land \neg
d\cap\mathsf{Cpre}(X_{8}))\,\cup\\
&\quad\quad\,\,\,\,(\neg a\land \neg b \land
c\land \neg d \cap\mathsf{Cpre}(X_{7}))\,\cup\\
&\quad\quad\,\,\,\,(a \land \neg b \land \neg d
\cap\mathsf{Cpre}(X_{5}))\cup(b \land
\neg d\cap\mathsf{Cpre}(X_{2}))\cup\\
&\qquad\,\,\,\,(d\cap\mathsf{Cpre}(X_{1}))
\end{align*}
\end{scriptsize}
\noindent Here, $\mathsf{Cpre}$ denotes the \emph{controllable predecessor} function; for a set $X\subseteq V$ of game nodes, $\mathsf{Cpre}(X)$ 
thus is the set of game nodes $v\in V$ such that the system player can enforce that $X$ is reached from $v$ by playing one step of the game.
This fixpoint equation system characterizes the winning region in EL games with objective formula $\varphi_{EL}$
(see~\cite{DBLP:conf/fossacs/HausmannLP24} for more details).
\end{example}

\noindent\textbf{Full proof of Lemma~\ref{lem:localMemory}}
\begin{proof}
We consider the case that $\mathbb{Q}_i=\exists$; the proof for the case that $\mathbb{Q}_i=\forall$ is analogous.
Let $\pi$ be an infinite run of $(D_i,F_i)$ on some infinite word $w\in (2^{AP})^\omega$.
By construction, every accepting state $q\in F_i$ is an accepting sink state (with $\delta_i(q,a)=q$ for all $a\in\Sigma$).
Thus $\pi$ visits some state in $F_i$ if and only if $\pi$ from some point on visits only states from $F_i$.
\end{proof}

\subsection*{Supplement to ``Solution of MP Games''}

\textbf{Full proof of Lemma~\ref{lemma:MPtoELcorrectness}}
\begin{proof}
We point out that plays $\pi'$ in $G_1$ have corresponding plays $\pi$ in $G$, obtained from $\pi'$ by removing
the auxiliary memory components. Conversely, plays $\pi$ in $G$ induce plays in $G_1$ that are obtained by annotating
$\pi$ with the auxiliary memory.

For one direction, let $\sigma$ be a winning strategy for the system player in $G$. Define a strategy $\sigma'$ for
the system player in $G_1$ by putting $\sigma'(\pi')=(\sigma(v),\mathsf{upd}{v,L})$ for any play
$\pi'$ in $G_1$ that ends in $(v,L)$ such that $v\in V_s$ and such that the associated play $\pi$ in $G$ is compatible
with $\sigma$. By definition of the memory updating function, $\mathsf{F}$-events stay in the auxiliary memory
after they have been visited once; similarly, $\mathsf{G}$-events are permanently removed from the memory
upon being violated once. Hence we have
\begin{align*}
\pi&\models\mathsf{F}\, c &\Leftrightarrow \pi'&\models\mathsf{F}\, c&\Leftrightarrow \pi'&\models\mathsf{GF}\, c &\Leftrightarrow \pi'&\models\mathsf{FG}\, c\\
\pi&\models\mathsf{G}\, c &\Leftrightarrow \pi'&\models\mathsf{G}\, c&\Leftrightarrow \pi'&\models\mathsf{GF}\, c &\Leftrightarrow \pi'&\models\mathsf{FG}\, c.
\end{align*}
It follows that for all plays $\pi'$ that are compatible with $\sigma'$, we have
$\pi\models\varphi$ if and only if $\pi'\models\varphi_1$, as required.

For the converse direction, let $\sigma'$ be a winning strategy for the system player in $G_1$. Define a strategy $\sigma$ for
the system player in $G$ that simply ignores the auxiliary memory values. We formally define this strategy as follows. Given a play 
$\pi$ in $G$ that ends in $v\in V_s$, assume that the induced play $\pi'$ is compatible with $\sigma'$, and assume that $\pi'$ ends in $(v,L)$.
Then put $\sigma(\pi)=w$, where $\sigma'(v,L)=(w,L')$.
Then $\sigma$ is a winning strategy for the same argument as in the proof of the previous direction.
Specifically, we have that for all plays $\pi$ that are compatible with $\sigma$, we have
$\pi\models\varphi$ if and only if $\pi'\models\varphi_1$.

\end{proof}

\noindent\textbf{Full proof of Lemma~\ref{lem:restr}}
\begin{proof}
Let $\pi=q_0 q_1\ldots\in\mathsf{plays}(A)$. 
The proof is by induction on $\varphi$. We consider just the nontrivial base cases
$\varphi=\mathsf{F}\,c$ and $\varphi=\mathsf{G}\,c$.
In both cases we have $\pi\models \varphi$ if and only if
$c\in \mathsf{ev}_\mathsf{F,G}(\pi)$ if and only if
$\varphi_{\mathsf{ev}_\mathsf{F,G}(\pi)}=\top$
if and only if 
$\pi\models\varphi_{\mathsf{ev}_\mathsf{F,G}(\pi)}$.
\end{proof}

\noindent\textbf{Full proof of Lemma~\ref{lemma:MP-to-EL-correctness}}
\begin{proof}
	For one direction,
	let $\sigma:V^*\cdot V_s\to V$ be a strategy for the system player
	in $G$ with which they win every play that starts at $v$.
	Define the system player strategy $\sigma':V'^*\cdot V'_s\to V'$ in $G_2$ as follows.
	Let $\pi'\in V'^*\cdot V'_s$ be a finite play in $G_2$ that starts at $(v,\mathsf{ev}_\mathsf{G})$ and ends in $(w,L)\in V'_s$.
	Then $\pi'$ induces the finite play $\pi$ in $G$ that is obtained from $\pi'$ by removing the memory values; we point out that $\pi$ starts
	in $v$ and ends in $w$.
	If $\pi$ is compatible with $\sigma$, then
	put $\sigma'(\pi')=(\sigma(\pi),\mathsf{upd}(L,w))$.
	To see that $\sigma'$ is a winning strategy,
	consider an infinite play $\pi'$ of $G_2$ that starts at $(v,\mathsf{ev}_{\mathsf{G}})$
	and that is compatible with $\sigma'$.
        We have to show $\pi'\models \varphi_L$ where
        $L=\mathsf{ev}_\mathsf{F,G}(\pi)$ denotes
	the set of $\mathsf{F}$- and $\mathsf{G}$-events satisfied
	by $\pi$.
        Let $\pi$ be the play in $G$ that is obtained from $\pi'$ by removing the auxiliary
        memory values.
        Since $\pi$ starts at $v$ and is compatible with $\sigma$
        by construction, the system player wins $\pi$ in $G$, that is, $\pi\models\varphi$.
        By Lemma~\ref{lem:restr}, we also have $\pi\models \varphi_{L}$.
	We observe that by definition of the memory updating function, $\pi'$ eventually reaches
	the subgame $G_L$ and then stays in this subgame forever. Hence $\pi'\models \mathsf{Inf}\,L$.
	We point out that $\gamma'$ outputs the local events from the memory component $L$ of nodes $(v,L)$ in $G'$ together with the Emerson-Lei events of the corresponding node $v$ in $G$      according to $\gamma$. Hence $\pi\models \varphi_{L}$ implies that $\pi'\models \varphi_L$.

	For the converse direction, let $\sigma':V'^*\cdot V'_s\to V'$ 
	be a strategy for the system player in $G_2$ such that they win
	every play that starts at $(v,\mathsf{ev}_{\mathsf{G}})$ and is 
	compatible with $\sigma'$. Define a strategy
	$\sigma:V^*\cdot V_s\to V$ for the system player in $G$ as follows.
	Let $\pi\in V^*\cdot V_s$ be a finite play that starts
	at $v$ and ends in some game node $w\in V_s$.
	Then $\pi$ induces a single play $\pi'$ in $G_2$ that starts
	at $(v,\mathsf{ev}_{\mathsf{G}})$ and ends in 
	$(w,L)$ for some set $L$ of events. If $\pi'$ is compatible
	with $\sigma'$, then put $\sigma(\pi)=w'$ where
	$\sigma'(\pi')=(w',L')$.
	To see that $\sigma$ is a winning strategy, let $\pi$ be
	an infinite play that is compatible with $\sigma$.
	By construction, $\pi$ induces an infinite play $\pi'$ of $G_2$
	that is compatible with $\sigma'$. Put $L=\mathsf{ev}_{\mathsf{F},\mathsf{G}}(\pi)$. 
	The play $\pi'$ eventually reaches the subgame $G_L$ and stays
	in this subgame forever. As the system
	player wins $\pi'$, we have $\pi'\models\mathsf{Inf}\,L\land \varphi_L$.
	By Lemma~\ref{lem:restr}, $\pi'\models\varphi$ which implies
	that $\pi\models\varphi$, as required.
\end{proof}

\paragraph{Obligation Games}

Consider the special case of MP games in which the objective of the game
is a Boolean
combination of $\mathsf{F}$- and $\mathsf{G}$-atoms.
This is the case for MP objectives without EL events
(that is, with $\Gamma_{\mathsf{EL}}=\emptyset$).
Applying to such games the reduction to compositional EL games described in Section~\ref{sec:ELreduction2}
yields DAGs of safety and reachability games.
In particular, in such games we have that for all $L\subseteq \Gamma_{\mathsf{F},\mathsf{G}}$, either $\varphi_L=\top$ or
$\varphi_L=\bot$. In the former case, the subgame
$G_L$ in the reduced DAG of EL games is a
safety game, in the latter case it is a reachability game.

Then the reduction desribed in Section~\ref{sec:ELreduction2} constructs DAGs of safety and reachability games and the upper bounds improve
as follows.
\begin{corollary}
	MP games without EL events, with $d$ local events and $m$ edges can be solved in time $\mathcal{O}(2^d \cdot m)$;
    winning strategies require at most $2^d$ memory values.
\end{corollary}

For MP games that have memory for local events, the time bound in the above Corollary improves to $\mathcal{O}(m)$ (the same as
for standard reachability or safety games); in this case, winning strategies are positional (that is, do not require memory).

If $\varphi_{\Gamma_\mathsf{F}}=\bot$, then the
bottom subgame $G_{\Gamma_\mathsf{F}}$ in the DAG is a reachability
game without escape, and hence the winning region of the system player
in this subgame is $\emptyset$.
In fact, the system player loses any game node in the reduced game
from which only nodes belonging to reachability subgames can be reached.
If $\varphi_{\Gamma_\mathsf{F}}=\top$, then the
bottom subgame $G_{\Gamma_\mathsf{F}}$ is a safety game without a way to violate the safety objective; then the system player
wins every node in this subgame. Again, the system player wins any game node in the reduced game
from which only nodes belonging to safety subgames can be reached.

\subsection*{Supplement to ``Synthesis via MP Games''}
\textbf{Full proof of Lemma~\ref{lem:DMPAconstruction}}
\begin{proof}
\begin{itemize}
\item[--] If $\mathbb{Q}_i=\exists$, then $\pi\models \mathbb{Q}_i\Phi_i$ iff there is a finite
prefix of $\pi$ that satisfies $\Phi_i$ which is the case iff there is a finite
prefix $\pi'$ of $\pi$ such that the run of the DFA $(D_i,F_i)$
on $\pi'$ ends in a state from $F_i$. This in turn is the case iff $\tau\models\varphi_i$ where
$\tau$ is the infinite run of $(D_i,O_i)$ on $\pi$.
\item[--] If $\mathbb{Q}_i=\forall$, then $\pi\models \mathbb{Q}_i\Phi_i$ iff all finite
prefixes of $\pi$ satisfy $\Phi_i$ which is the case iff for each finite 
prefix $\pi'$ of $\pi$, the run of the DFA $(D_i,F_i)$
on $\pi'$ ends in a state from $F_i$. This in turn is the case iff $\tau\models\varphi_i$ where
$\tau$ is the infinite run of $(D_i,O_i)$ on $\pi$.
\item[--] If $\mathbb{Q}_i=\forall\exists$, then $\pi\models \mathbb{Q}_i\Phi_i$ iff there are infinitely many finite
prefixes of $\pi$ satisfy $\Phi_i$ which is the case iff there are infinitely many finite 
prefixes $\pi'$ of $\pi$ such that the run of the DFA $(D_i,F_i)$
on $\pi'$ ends in a state from $F_i$. This in turn is the case iff $\tau\models\varphi_i$ where
$\tau$ is the infinite run of $(D_i,O_i)$ on $\pi$.
\item[--] If $\mathbb{Q}_i=\exists\forall$, then $\pi\models \mathbb{Q}_i\Phi_i$ iff all but finitely many finite
prefixes of $\pi$ satisfy $\Phi_i$ which is the case iff for all but finitely many finite 
prefixes $\pi'$ of $\pi$, the run of the DFA $(D_i,F_i)$
on $\pi'$ ends in a state from $F_i$. This in turn is the case iff $\tau\models\varphi_i$ where
$\tau$ is the infinite run of $(D_i,O_i)$ on $\pi$.
\end{itemize}
\end{proof}

\noindent\textbf{Full proof of Proposition~\ref{prop:formulaToDMPA}}
\begin{proof}
We point out that MP objectives natively encode intersection and disjunction of languages as follows:
Given two disjoint deterministic transition systems $D_1, D_2$ and two MP objectives $O_1=(\mathsf{ev}_1,\gamma_1,\varphi_1)$ and
$O_2=(\mathsf{ev}_2,\gamma_2,\varphi_2)$ such that $\mathsf{ev}_1\cap \mathsf{ev}_2=\emptyset$, 
put $\mathsf{ev}=\mathsf{ev}_1\cup\mathsf{ev}_2$ and define
$\gamma:D_1\times D_2\to 2^\mathsf{ev}$ by $\gamma(q_1,q_2)=\gamma_1(q_1)\cup\gamma_2(q_2)$. Then
\begin{itemize}
\item $L(D_1,O_1)\cap L(D_2,O_2) = L(D_1\times D_2,(\mathsf{ev},\gamma,\varphi_1\land\varphi_2)$;
\item $L(D_1,O_1)\cup L(D_2,O_2) = L(D_1\times D_2,(\mathsf{ev},\gamma,\varphi_1\lor\varphi_2)$.
\end{itemize}
Then the claim follows by induction from the correctness of the individual finite trace automata 
(Lemma~\ref{lem:DMPAconstruction}).
\end{proof}

\end{document}